\documentclass[12pt, twoside]{amsart}
\usepackage{secdot}
\scrollmode 
\usepackage{amsmath}
\usepackage{amsthm}
\usepackage{amssymb}
\usepackage{amsfonts}
\usepackage{dsfont}
\usepackage{mathrsfs}
\usepackage{hyperref}
\usepackage[nameinlink]{cleveref}
\usepackage[matrix,arrow,curve]{xy}
\usepackage{tikz}
\usepackage{titlesec}
\usepackage{mathtools}
\usetikzlibrary{arrows.meta}

\newtheorem{theorem}{Theorem}[section]

\newtheorem{proposition}[theorem]{Proposition}
\newtheorem{corollary}[theorem]{Corollary}
\theoremstyle{definition}
\newtheorem{definition}[theorem]{Definition}

\newtheorem{conjecture}[theorem]{Conjecture}
\newtheorem{remark}[theorem]{Remark}

\newcommand{\Lie}{%
  \mathchoice
    {\ooalign{$\displaystyle\mathcal{L}$\cr\hidewidth\rule[0.9ex]{0.6em}{0.6pt}\hidewidth\cr}}
    {\ooalign{$\textstyle\mathcal{L}$\cr\hidewidth\rule[0.8ex]{0.55em}{0.6pt}\hidewidth\cr}}
    {\ooalign{$\scriptstyle\mathcal{L}$\cr\hidewidth\rule[0.6ex]{0.4em}{0.5pt}\hidewidth\cr}}
    {\ooalign{$\scriptscriptstyle\mathcal{L}$\cr\hidewidth\rule[0.5ex]{0.35em}{0.5pt}\hidewidth\cr}}
}

\crefname{section}{Section}{sections}
\crefname{figure}{Figure}{figures}
\crefname{definition}{Definition}{definitions}
\crefname{proposition}{Proposition}{propositions}
\crefname{lemma}{Lemma}{lemmas}
\crefname{theorem}{Theorem}{theorems}
\crefname{conjecture}{Conjecture}{conjectures}
\crefname{corollary}{Corollary}{corollaries}
\crefname{remark}{Remark}{remarks}
\crefname{equation}{Eq.}{Eqs.}

\newcommand{\bcref}[1]{\textbf{\cref{#1}}}
\newcommand{\icref}[1]{\textit{\cref{#1}}}

\definecolor{softred}{RGB}{170,70,60}
\definecolor{nightblue}{RGB}{25,25,112}

\titleformat{\subsection}
  {\normalfont\scshape} 
  {\thesubsection}{1em}{}

\begin{document}

\title[Integrability in the $\text{BMS}_3$ scenario]
{Integrability in Asymptotic Symmetries of Spacetime: The $\text{BMS}_3$ scenario}

\author{Corentin Vitel}
\address{Université Bourgogne Europe, CNRS, IMB UMR 5584, F-21000 Dijon, France}
\email{corentin.vitel@u-bourgogne.fr}

\maketitle

\begin{abstract}
We discuss the existence of a well-defined $\mathfrak{bms}_3$-integrable hierarchy, introduced in \cite{Fue.al18}, using different structural methodologies. Specifically, from the variational complex on the ring of polynomial symbols, a $\mathfrak{bms}_3$ bi-Hamiltonian structure is constructed on which a Nijenhuis operator can be attached. A similar argument is made for the $\mathrm{AdS}_3$ case, where we check that the flat limit of the latter recovers the asymptotically flat situation. An alternative $\tau$-scheme description is also presented. Moreover, a Lie-Poisson description suggests that such a $\mathfrak{bms}_3$-hierarchy is not unique. For a subclass of so-called energy-dependent Schrödinger operators, it is shown that their Lax flows are described by the coadjoint orbits of $\mathfrak{bms}_3$. 
\end{abstract}

\section*{Introduction}\label{sec:introduction}
\phantomsection

The setting of the present study is to characterize in a systematic manner an integrable system built upon the centrally extended $\mathfrak{bms}_3$ Lie algebra, using several complementary approaches to integrability. 
In particular, we revisit the work of \cite{Fue.al18}, the present systematic description leading to the identification of new integrable structures based on the same Lie algebra, thus extending the original results. In the spirit of the latter, a natural application of this mathematical study lies in the context of integrability in $2+1$-dimensional gravity.
Aiming at bridging the gap between the two communities, we review in the following some general elements of integrability and gravitational asymptotic symmetries. This emphasizes on the two perspectives of the problem. 

\subsection*{A perspective from integrability}

By \textit{complete integrability}, one commonly means a dynamical system with sufficiently many independent conserved quantities. 
Given a symplectic $2n$-dimensional manifold with some Poisson tensor associated with the symplectic form, we characterize \textit{integrability} as follows: there exists $n$ independent functionals that are mutually commuting with respect to the bracket attached to this Poisson tensor. 

Even if the story of \textit{completely} integrable systems dates back to Liouville, Jacobi, Euler and many others, it is only in the late 1960s that \cite{Gar.al67} discussed an infinite-dimensional analog of the known features: the \textit{Korteweg-de Vries} (KdV) \textit{hierarchy}. Then, \cite{Lax68,ZakFad71,Gar71,Mag78} laid the groundwork for a first proper characterization of an \textit{infinite-dimensional integrable system}. Indeed, compared to the previous finite dimensional setting, we may expect things to turn sour, \textit{a priori}; except for some very specific systems, for which a sufficiently robust definition of
\textit{integrability} can be given. 

Roughly speaking, we would like to exhibit an \textit{infinite} number of \textit{mutually commuting conserved quantities} $\{\widetilde{\mathcal{H}}_n\}_{n\geq0}$ attached to \textit{infinitely} many \textit{independent} flows $\big\{\overline{\mathcal{K}}_n\big\}_{n\geq0}$. These flows give rise to a set of non-linear Partial Differential Equations (PDEs), $\partial\cdot /\partial t=\overline{\mathcal{K}}_n[\cdot]$, often referred as a \textit{hierarchy}. Equivalently, we could say that each $\overline{\mathcal{K}}_n$ is an infinitesimal generator for a transformation in time $t$ that acts on some functions $f_m$ through the \textit{evolution equation} $\partial f_m/\partial t=\overline{\mathcal{K}}_n[f_m]$ (\cite{MiwJimDat00}). 
By \textit{independent}, we mean that all the $\overline{\mathcal{K}}_n$ mutually commute.

In the spirit of \cite{Sem02}, one should keep in mind that these PDEs are \textit{Hamiltonian} with respect to (at least) a Poisson bracket (see the pioneer works of \cite{GelDic76,GelDic76-2,GelManShu76}), and stand as compatibility conditions for a set of linear equations. We expect the $\{\widetilde{\mathcal{H}}_n\}_{n\geq0}$ to commute with respect to such bracket(s). The flows $\{\overline{\mathcal{K}}_n\}_{n\geq0}$ are then defined from a so-called \textit{Hamiltonian operator}: an infinite-dimensional analogue of the Poisson tensor whose bracket is precisely the one aforementioned.  

From a geometrical perspective, the trajectories of such systems are described by the coadjoint orbits of an infinite-dimensional Lie group. Thinking once again about the KdV situation, it has been well-known how to describe its successive flows from the coadjoint action of the \textit{Virasoro algebra} (see for instance \cite{Seg81,Seg91}).

Infinite-dimensional integrability has informed a huge number of areas in mathematical physics among which 2D topological gravity is a case in point. To illustrate, recall that \cite{Wit91} conjectured that the intersection theory on the moduli space of curves was somehow governed by the KdV hierarchy. Straight afterwards, \cite{Kon92} proved the conjecture; the KdV hierarchy being at the core of the argument. 

All in all, here is an attempt at a formal definition that summarizes these features: 

\begin{definition}\label{def:def0.1}
    An \textit{infinite-dimensional integrable system} is a set of infinitely many independent flows associated with infinitely many commuting conserved quantities. 
\end{definition}

\subsection*{A perspective from gravitational asymptotic symmetries}

The natural context for discussing spacetime asymptotic symmetries is the study of isolated systems in General Relativity. The latter can be modeled in terms of \textit{asymptotically flat spacetimes} \cite{Pen63}. 
When moving away from such systems and reaching the far-away `flat' region, one would expect to recover the isometries of the Minkowski (flat) spacetime, that is the Poincaré group. 
However, it happens that the group respecting the asymptotic structure of these spacetimes indeed includes Poincaré, but is actually larger: namely, the infinite-dimensional \textit{BMS group}.

Originally due to the work of Bondi, Metzner, Sachs and van der Burg \cite{BonBurMet62, Sac62, Sac62-2}, the \textit{BMS group} then brings together the diffeomorphisms that preserve the asymptotic form of the spacetime metric. Since the early 2010s and the pioneering works of Barnich and Troessaert (non-exhaustively \cite{BarTro10, BarTro10-2, BarTro11}), these symmetries have shed light on numerous issues, from soft theorems and gravitational-wave memory effects \cite{Str14,StrZhi16} to quantum gravity and holography \cite{BroHen86,Mal99,PasShaStr17}, and other related black hole physics problems \cite{Haw76, HawPerStr16, HawPerStr17}.

In the $3$-dimensional case, the \textit{BMS$_3$ algebra} has a specific semi-direct product structure involving two copies of the Virasoro algebra, the second being treated as an Abelian ideal. It is well-known that it emerges as the limit of the \textit{$2$-dimensional conformal algebra} \cite{Bag10}.

Thus, while on the one hand, relations between Conformal Field Theories (CFTs) and integrable systems are well documented \cite{DiFMatSén97, Zub91} (cf. also the role of \textit{W-algebras} \cite{Dic97} and \textit{Poisson Vertex Algebras} \cite{Bel89,BarDeSKac09,Kac17}); and while the $3$-dimensional version of the BMS algebra is built from the Virasoro algebra, one may wonder whether, in accordance with \textbf{Definition~\ref{def:def0.1}}, there exists some `$\mathfrak{bms}_3$ integrable system'.

\subsection*{Where Integrability meets Asymptotic symmetries: the bms$_3$ algebra}

\cite{PérTemTro16} originally showed that an appropriate choice of boundary conditions on an AdS$_3$-spacetime reduces the Einstein equations to the PDEs describing the KdV hierarchy. At that time, they were already wondering about the relevance of the choice made to fix these boundary conditions in their construction. As hinted, it was suggested that such a choice was not unique. Nevertheless, it seemed natural to try a similar construction for a $3$-dimensional spacetime free of cosmological constant. 

To our knowledge, \cite{Fue.al18} were the first to tackle and answer this question in the affirmative. Starting with the natural BMS$_3$ Poisson structure, they built a hierarchy based on a $2$-dimensional $\mathrm{ISL}(2,\mathbb{R})$-valued gauge field with right choices of boundary conditions. As before, they then described how it was possible to reduce the Einstein equations in vacuum to the PDEs associated with this hierarchy, thus strengthening the connection previously glimpsed.  

In this context, we adopt a more `integrability-oriented' path, in the sense that we investigate the existence of an infinite-dimensional integrable system purely based on the semi-direct algebra structure of $\mathfrak{bms}_3$. 

In a broader perspective, our aim is to connect classes of boundary conditions to $\mathfrak{bms}_3$-like bi-Hamiltonian (and hopefully integrable) systems, so that each class would correspond to a reduction of Einstein equations to the PDEs of a given hierarchy. Ultimately, this is in the spirit of understanding the various, yet subtle, interrelations between Integrable Systems and General Relativity.

\subsection*{Main Results}

In this paper, we present a systematic approach to build a well-defined $\mathfrak{bms}_3$ bi-Hamiltonian hierarchy by exploiting its Lie algebra structure: \textit{i)} either by working on the variational complex in the ring of polynomial symbols; \textit{ii)} or from  the $\mathfrak{bms}_3$ Lie-Poisson bracket with an appropriate `frozen' partner. We make sure that this hierarchy is integrable. Crucially, the second approach hints that such a hierarchy is not unique. In a similar fashion, we detail the construction of an integrable hierarchy based on the $\widehat{\mathfrak{ads}}_3$ algebra. Moreover, we show that the $\mathfrak{bms}_3$ coadjoint orbits describe the Lax phase space for some Schrödinger-like spectral problems. 

\subsection*{Plan of the paper and Notations}

In \bcref{sec:section1}, we give a brief summary of the main results and ideas concerning the (extended) BMS$_3$ group and its (extended) algebra. We emphasize in its connection with the Virasoro algebra and conclude by recalling the coadjoint action on its regular dual. In \bcref{sec:section2}, we construct a bi-Hamiltonian $\mathfrak{bms}_3$ hierarchy; we explicitly state its first flows and (conserved) functionals; we prove it is integrable. We conclude with an alternative $\tau$-scheme formulation in the asymptotically flat situation. The process is repeated in \bcref{sec:section3} on $\widehat{\mathfrak{ads}}_3$. We pay attention to how the flat limit allows one to recover one bi-Hamiltonian structure from the other. \bcref{sec:section4} recovers the previous results in a `geometric/Euler-oriented' perspective. In particular, it is suggested how extra $\mathfrak{bms}_3$-like bi-Hamiltonian hierarchies should exist. Finally, \bcref{sec:section5} examines the link between this integrable system and energy-dependent Schrödinger operators. 
We believe that the next schematic figure may be a useful tool for the reader to help visualize the structure of the paper.

%Throughout this paper, a few notations may seem unusual. We clarify their meaning. 
%By $\partial$, we denote the total derivative $d/dx$, and $(d/dx)^if$ by $f^{(i)}$. Similarly, $\partial_x$ refers to $\partial/\partial x$, $\delta_u$ is an abridged version of the variational derivative $\delta/\delta u$ and $\Lie_{X}$ stands for the Lie derivative (in the usual sense of the Cartan formula).
%The `$\propto$\hspace{0.1cm}' symbolizes `proportional to' and the letter $G$ represents the gravitational constant. 

%Lastly, we believe that the next schematic figure may be a useful tool for the reader to help visualize the structure of the paper. 

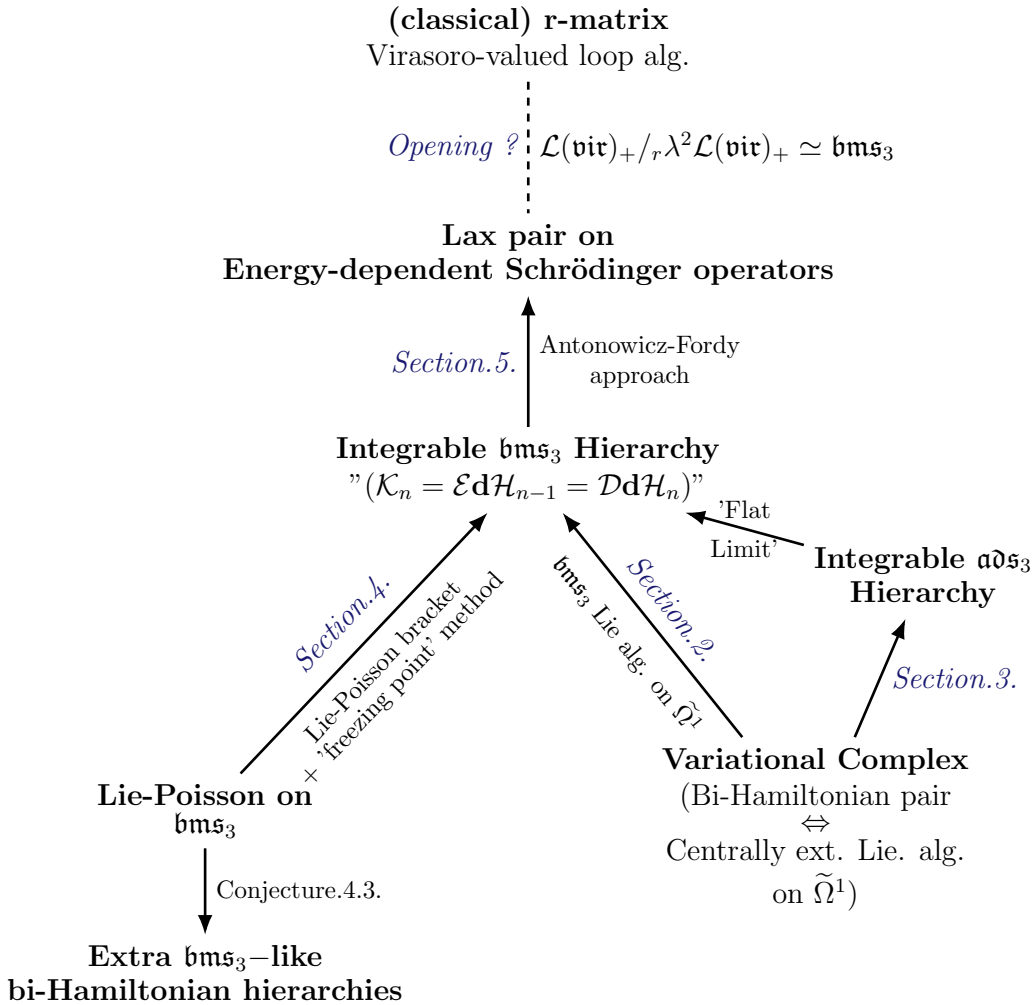
\begin{figure}[ht]
\begin{center}
\begin{tikzpicture}[
    scale=.7,
    transform shape,
    >=Latex,
    every node/.style={font=\normalsize}
]

% Nodes

\node (A) at (-1.5,6)
{\shortstack{{\textbf{(classical) r-matrix}}\\ Virasoro-valued loop alg.}};

\node (A') at (-1.5,3)
{\shortstack{\textbf{Lax pair on}\\\textbf{Energy-dependent Schrödinger operators}       }};

\node (B) at (-6,-4.75)
{\shortstack{\textbf{Lie-Poisson on}\\
$\mathfrak{bms}_3$}};

\node (C) at (2.5,-5)
{\shortstack{\textbf{Variational Complex}\\(Bi-Hamiltonian pair\\ $\Leftrightarrow$ \\ Centrally ext. Lie. alg.\\ on $\widetilde{\Omega}^1$)}};

\node (D) at (4,-1.5)
{\shortstack{\textbf{Integrable $\widehat{\mathfrak{ads}}_3$}\\\textbf{Hierarchy}}};

\node (E) at (-6,-7)
{\shortstack{\textbf{Extra $\mathfrak{bms}_3$-like} \\ \textbf{bi-Hamiltonian hierarchies}}};

% Centre

\node (S) at (-1.5,0)
{\shortstack{\textbf{Integrable $\mathfrak{bms}_3$ Hierarchy}\\
"$(\mathcal{K}_n=\mathscr{E}\mathbf{d}\mathcal{H}_{n-1}=\mathscr{D}\mathbf{d}\mathcal{H}_n)$"}};

% Arrows

\draw[line width=1pt,dashed] (A)--(A') 
    node[midway, left]{\textit{\textcolor{nightblue}{Opening ?}}}
    node[midway, right]{$\mathcal{L}(\mathfrak{vir})_+/_r\lambda^2\mathcal{L}(\mathfrak{vir})_+\simeq\mathfrak{bms}_3$};

\draw[line width=1pt, ->] (S)--(A')
    node[midway, left]{\textit{\textcolor{nightblue}{Section 5}}}
    node[midway, right]{\footnotesize\shortstack{Antonowicz-Fordy\\approach}};

\draw[,line width=1pt,->] (B)--(S)
        node[midway, sloped, above, yshift=2pt]{\textit{\textcolor{nightblue}{Section 4}}}
        node[midway, sloped, below, yshift=-2pt]{\footnotesize\shortstack{Lie-Poisson bracket\\ + 'freezing point' method}};

\draw[line width=1pt,->] (C)--(S) 
    node[midway, sloped, above, yshift=2pt]{\textit{\textcolor{nightblue}{Section 2}}}
    node[midway, sloped, below, yshift=-2pt] {\footnotesize $\mathfrak{bms}_3$ Lie alg. on $\widetilde{\Omega}^1$};

\draw[line width=1pt, ->] (C)--(D)
    node[midway, right]{\textit{\textcolor{nightblue}{Section 3}}};

\draw[line width=1pt, ->] (D)--(S)
    node[midway, above]{\footnotesize'Flat}
    node[midway, below]{\footnotesize Limit'};

\draw[line width=1pt, ->] (B)--(E)
    node[midway, right]{\footnotesize Conjecture 4.4.};

\end{tikzpicture}
\end{center}
\caption{Overview of the structures involved in the integrable $\mathfrak{bms}_3$ hierarchy.}
\label{fig:bms3_hierarchy}
\end{figure}

Throughout this paper, a few notations may seem unusual. We clarify their meaning. 
By $\partial$, we denote the total derivative $d/dx$, and $(d/dx)^if$ by $f^{(i)}$. Similarly, $\partial_x$ refers to $\partial/\partial x$, $\delta_u$ is an abridged version of the variational derivative $\delta/\delta u$ and $\Lie_{X}$ stands for the Lie derivative (in the usual sense of the Cartan formula).
The `$\propto$\hspace{0.1cm}' symbolizes `proportional to' and the letter $G$ represents the gravitational constant.

\section{BMS$_3$: From its physical intuition to its geometrical description}\label{sec:section1}
After reviewing general results about exceptional semi-direct products, we comment on the particular case of the $\text{BMS}_3$ group. In particular, we recall the expression for the coadjoint action of its centrally extended algebra. The latter had been studied mainly in the context of $\text{BMS}_3$ particles \cite{Obl18, BarObl15}.

\subsection{Basics on exceptional semi-direct products}\label{sec:section1.1}

Consider a semi-direct product structure $H:=G\ltimes_{\sigma}A$, for a Lie group $G$, a vector group $A$ (a vector space endowed with the addition as the group operation) and a group homomorphism $\sigma:G\to\text{GL}(A)$, $g\mapsto\sigma_g$. Denote the elements of $H$ by the pair $(g,\alpha)$. The group operation $\circ$ writes 
\begin{equation}
    (g_1,\alpha_1)\circ(g_2,\alpha_2):=(g_1.g_2, \alpha_1+\sigma_{g_1}\alpha_2)\hspace{0.1cm}.
    \label{Eq.1}
\end{equation}

Regarding $\sigma$ at the origin $e\in G$, one considers the standard Lie algebra homomorphism $\text{d}_e\sigma:\mathfrak{g}\to\text{End}(A)$, $X\mapsto\big(\text{d}_e\sigma\big)_X$, where for $\alpha\in A$, $\big(\text{d}_e\sigma\big)_X\alpha=\frac{d}{dt}\big(\sigma_{e^{tX}}\alpha\big)\big|_{t=0}$. The latter endows the Lie algebra $\mathfrak{h}$ with a semi-direct structure
\begin{equation}
    \mathfrak{h}=\mathfrak{g}\ltimes_{\text{d}_e\sigma}A, \hspace{0.2cm} (X,\alpha)\in\mathfrak{h}\simeq\mathfrak{g}\oplus A\hspace{0.1cm}.
    \label{eq.2}
\end{equation}

By \textit{exceptional semi-direct product} \cite{BarObl14, Raw75, Bag98}, we mean the specific situation where $A\equiv\mathfrak{g}_{\text{ab}}$ the Lie algebra of $G$, in the sense of an abelian additive vector space. The group homomorphism then reduces to $\sigma\equiv\text{Ad}$ with an attached Lie algebra homomorphism $\text{d}_e\sigma\equiv\text{ad}$, respectively the adjoint representation of $G$ acting on $\mathfrak{g}$ and its infinitesimal version. 

One shows that the adjoint representation of the Lie algebra $\mathfrak{h}$ is given by 
\begin{equation}
    \text{ad}_\mathfrak{h}\big(X_1,\alpha_1\big)\cdot\big(X_2,\alpha_2\big)
    =\big(\text{ad}_\mathfrak{g}X_1\cdot X_2,\hspace{0.1cm} \text{ad}_\mathfrak{g}X_1\cdot\alpha_2-\text{ad}_\mathfrak{g}X_2\cdot\alpha_1\big)\hspace{0.1cm},
    \label{eq.3}
\end{equation}
and the coadjoint action of $\mathfrak{h}$ on its dual $\big(j,p\big)\in\mathfrak{h}^*\simeq\mathfrak{g}^*\oplus\mathfrak{g}^*_{\text{ab}}$ is
\begin{equation}
    \text{ad}_\mathfrak{h}^*\big(X,\alpha\big)\cdot\big(j,p\big)=\big(\text{ad}_\mathfrak{g}^*X\cdot j+\text{ad}_\mathfrak{g}^*\alpha\cdot p, \hspace{0.1cm}\text{ad}^*_\mathfrak{g}X\cdot p\big)\hspace{0.1cm}.
    \label{eq.4}
\end{equation}

\subsection{BMS$_3$ stemming from Poincaré}\label{sec:section1.2}
\textit{Gravitational asymptotic symmetries} are diffeomorphisms that leave invariant the asymptotic behavior of the gravitational field. 

On a $3\text{d}-$Minkowski spacetime, one commonly works with the so-called \textit{(retarded) Bondi coordinates}, namely a triplet $(u,r,x)$, where `$u=\text{constant}$' labels outgoing null hypersurfaces spanned by outgoing light rays, $r$ is a radial coordinate on these hypersurfaces, and $x$ denotes the angular coordinate of a corresponding sphere $S^1$.  

The region reached when $r\to+\infty$ has the topology of $S^1\times\mathbb{R}$, one (future) circle for each value of $u$. We name this cylinder the \textit{(future) null infinity} $\mathscr{I}^+$.

The \textit{BMS$_3$} group may be seen as a `double' generalization of the \textit{Poincaré} group, in a sense clarified below. The latter is the isometry group of the $\mathrm{d}$-dim Minkowski spacetime.  Its mathematical description lies on the semi-direct group structure \cite{Obl18-2}
\begin{equation}
    \text{ISO}(\mathrm{d}-1,1)^\uparrow=\text{SO}(\mathrm{d}-1,1)^\uparrow\ltimes\mathbb{R}^\mathrm{d}\hspace{0.1cm},
    \label{eq.5}
\end{equation}
whose \textit{Lorentz} transformations $f$ (technically speaking, the proper orthochronous ones) act on the space-time translations $\alpha$ in the standard way $(f_1,\alpha_1)\cdot(f_2,\alpha_2)=(f_1\cdot f_2,\alpha_1+f_1\cdot\alpha_2)$, the action of the right-hand side being understood as the usual matrix multiplication.

In the three-dimensional case, since
\begin{equation}
    \text{SO}(2,1)^\uparrow\simeq \text{SL}(2,\mathbb{R})/\{\pm 1\}\hspace{0.1cm},
    \label{eq.6}
\end{equation}
it is common to confuse the $3\text{d}$-Poincaré group with its double cover. The action happens to be precisely the adjoint action on $\text{SL}(2,\mathbb{R})$, turning the group into an \textit{exceptional semi-direct product}:
\begin{equation}
    \text{ISO}(2,1)^\uparrow\equiv\text{SL}(2,\mathbb{R})\ltimes_\text{Ad}\mathfrak{sl}(2,\mathbb{R})_{(\text{ab})}\hspace{0.1cm}.
    \label{eq.7}
\end{equation}

Bondi coordinates are suited such that the Poincaré transformations remain well-defined at $\mathscr{I}^+$ \cite{Obl18-2}. There, the general action of the group is
\begin{equation}
    (u,x)\longmapsto\big(f^{(1)}(x)u+\alpha(f(x)),f(x)\big)\hspace{0.1cm},
    \label{eq.8}
\end{equation}
where $f$ represents the projective transformation on the circle and $\alpha$ a space-time translation.

BMS$_3$ generates diffeomorphisms that preserve the asymptotic behavior of the Minkowski metric\footnote{Actually, the \textit{BMS}$_3$ group is the quotient of such diffeomorphisms by those whose associated surface charges are trivial.}. That is, we look for transformations (in Bondi coordinates) that conserve the \textit{fall-off} conditions of a $3\text{d}$-\textit{asymptotically flat spacetime} \cite{Obl18}:
\begin{equation}
    ds^2\stackrel{r\to+\infty}{\sim}\mathcal{O}(1)du^2-\big(2+\mathcal{O}(1/r)\big)dudr+r^2dx^2+\mathcal{O}(1)dudx\hspace{0.1cm}.
    \label{eq.9}    
\end{equation} 
At the algebraic level, denoting by $\Xi_{(X,\alpha)}$ the \textit{asymptotic Killing vectors}, their Lie bracket satisfies \cite{BarCom07}: 
\begin{equation}
    \big[\Xi_{(X_1,\alpha_1)},\Xi_{(X_2,\alpha_2)} \big]=\Xi_{\big([X_1,X_2],\hspace{0.1cm} [X_1, \alpha_2]-[X_2,\alpha_1] \big)}+\mathcal{O}\hspace{0.1cm},
    \label{eq.10}
\end{equation}
where $\mathcal{O}$ refers to neglected terms, in the sense that they have trivial surface charges. With a Fourier mode expansion $X_{m}\equiv\Xi_{(e^{imx},0)}$ and $\alpha_{m}\equiv\Xi_{(0,e^{imx})}$, one recovers the previous result \cite{BarCom07}: 
\begin{equation}
    i[X_{m},X_{n}]=(m-n)X_{m+n}, \hspace{0.2cm} i[X_{m},\alpha_{n}]=(m-n)\alpha_{m+n}, \hspace{0.2cm} i[\alpha_{m},\alpha_{n}]=0, \hspace{0.2cm} \forall m,n\in\mathbb{Z}\hspace{0.1cm}.
    \label{eq.11}
\end{equation}

As for Poincaré, the BMS$_3$ group also carries a semi-direct structure that spans a Witt-like algebra, extending the Lorentzian part and acting on an infinite-dimensional abelian algebra. At the group level, BMS$_3$ transformations at $\mathscr{I}^+$ generalize \icref{eq.8} where now $f\in\text{Diff}(S^1)$ and $\alpha\in C^\infty(S^1)$. We call them respectively \textit{superrotations} \cite{BarTro11-2} and \textit{supertranslations} \cite{Sac62}. In this sense, following \cite{PriSch22, Ruz20}, we are considering \textit{extended} BMS$_3$ actions.
\\
Again, similarly to the Poincaré situation with the Lorentz group, we slightly abuse notation. By $\text{Diff}(S^1)$ we actually understand $\widetilde{\text{Diff}}^+(S^1)$, the universal cover of the orientation-preserving diffeomorphisms group. 

A closer look at the group operation allows to deduce that BMS$_3$ still possesses an exceptional semi-direct structure:
\begin{equation}
    (u,x)\xmapsto{(f_1,\alpha_1)\circ(f_2,\alpha_2)}\bigg((f_1\circ f_2)^{(1)}(x)u+\big[\alpha_1+\text{Ad}f_1\cdot\alpha_2\big](f_1\circ f_2)(x), \hspace{0.1cm} (f_1\circ f_2)(x)\bigg)\hspace{0.1cm}.
    \label{eq.12}
\end{equation}

In the above expression, $\text{Ad}$ denotes the adjoint action of diffeomorphisms on vector fields. Rather than being seen as smooth functions on the circle, supertranslations are taken as $(-1)$-densities $\alpha=\alpha(x)dx^{-1}$ attached to an angle-dependent translation at given $u$.

Therefore, we tackle BMS$_3$ as the $\text{Diff}(S^1)$-exceptional semi-direct group:
\begin{equation}
    \text{BMS}_3:=\text{Diff}(S^1)\ltimes_{\text{Ad}}\text{Vect}(S^1)_{(\text{ab)}}\hspace{0.1cm}.
    \label{eq.13}
\end{equation}

It is a well-defined infinite-dimensional Lie-Fréchet group \cite{KriMic98} (or \cite{KriMic97} for an exhaustive treatment) and admits a universal central extension \cite{Obl18}.

\subsection{Centrally extended $\mathfrak{bms}_3$ and coadjoint action}\label{sec:section1.3}

Let $\text{Vect}(S^1)=\text{Lie}\big(\text{Diff}(S^1)\big)$ be the Lie algebra of vector fields on the circle. Its elements write $X=X(x)\partial_x$, $X$ being smooth, and the Lie bracket is $\big[X_1,X_2\big]_{\text{Vect}(S^1)}=\big(X_1X_2^{(1)}-X_2X_1^{(1)}\big)\partial_x$. 

$\mathcal{H}^2\big(\text{Vect}(S^1)\big)$ is one-dimensional and is generated by the \textit{Gel'fand-Fuks} cocycle \cite{GuiRog07}
\begin{equation}
    \omega(X_1,X_2)=\int_{S^1}dxX_1^{(3)}X_2\hspace{0.1cm}.
    \label{eq.14}
\end{equation}

The \textit{Virasoro} algebra  corresponds to the universal central extension of $\text{Vect}(S^1)$ by means of this cocycle
\begin{equation}
\begin{aligned}
    \mathfrak{vir}:=\widehat{\text{Vect}}(S^1)&=\text{Vect}(S^1)\oplus\mathbb{R}\hspace{0.1cm},\\ \hspace{0.1cm} \big[(X_1,a),(X_2,b)\big]_{\mathfrak{vir}}&=\big(\big[X_1,X_2\big]_{\text{Vect}(S^1)},\omega(X_1,X_2)\big)\hspace{0.1cm}.
    \label{eq.15}
\end{aligned}
\end{equation}

\begin{definition}\label{def:definition1.1}
    The centrally extended Lie algebra $\widehat{\mathfrak{bms}_3}$ is the Virasoro exceptional semi-direct sum
    \begin{equation}
    \widehat{\mathfrak{bms}}_3:=\widehat{\text{Vect}}(S^1)\ltimes_{\widehat{\text{ad}}}\widehat{\text{Vect}}(S^1)_{\text{ab}}, \hspace{0.2cm} (X,a;\hspace{0.1cm} \alpha,b)\in\widehat{\mathfrak{bms}}_3\simeq\mathfrak{vir}\oplus\mathfrak{vir}_{\text{ab}}\hspace{0.1cm},
    \label{eq.16}
\end{equation}
where $X$ and $\alpha$ are vector fields, $a$ and $b$ real numbers. Its Lie bracket stems from the one of the Virasoro algebra:
\begin{equation}
    \begin{split}
        &\big[(X_1,a_1;\hspace{0.1cm} \alpha_1,b_1),(X_2,a_2;\hspace{0.1cm} \alpha_2,b_2) \big]_{\widehat{\mathfrak{bms}_3}}\\
        &=\bigg(\big[X_1,X_2\big],\omega(X_1,X_2);\hspace{0.1cm}\big[X_1,\alpha_2 \big]-\big[X_2,\alpha_1 \big], \omega(X_1,\alpha_2)-\omega(X_2,\alpha_1) \bigg)\hspace{0.1cm},
    \end{split}
    \label{eq.17}
\end{equation}
being understood that the brackets on the right-hand side correspond to those on $\text{Vect}(S^1)$.
\end{definition}

Following the usual terminology, we say that $X$ represents an \textit{infinitesimal superrotation} and $\alpha$ an \textit{infinitesimal supertranslation}.

Except when there will be a risk of confusion with the notation, from now on, we deliberately mix $\mathfrak{bms}_3$ with its extended version.

Call $\mathcal{F}_w(S^1)$ the infinite-dimensional vector space of densities of weight $w\in\mathbb{R}$. It is enough (\cite{Kir89}) to consider its \textit{regular dual} $\mathcal{F}^*_{w,\text{reg}}(S^1)$ for regular distributions. Thanks to the isomorphism (\cite{Kir89, GuiRog07}) $\mathcal{F}^*_{w,\text{reg}}(S^1)\simeq\mathcal{F}_{1-w}(S^1)$, elements dual to vector fields are quadratic densities. Thus, we identify
\begin{equation}
    \big(udx^2,c^*\big)\in\mathfrak{vir}^*_{\text{reg}}\simeq\mathcal{F}_2(S^1)\oplus\mathbb{R}\hspace{0.1cm}.
    \label{eq.18}
\end{equation}

Since the space of coadjoint vectors of $\mathfrak{g}\oplus A$ is the dual of this semi-direct sum whose elements are pairs $(j,p)$ with the standard pairing $\langle(j,p),(X,\alpha)\rangle=\langle j,X\rangle_{\mathfrak{g}}+\langle p,\alpha\rangle_A$, the regular dual of $\mathfrak{bms}_3$ again is deduced from the one of Virasoro: 
\begin{equation}
    (j\cdot dx^2,c_1;\hspace{0.1cm}p\cdot dx^2,c_2)\in\mathfrak{bms}^*_{3_{\text{reg}}}\simeq\mathfrak{vir}^*_\text{reg}\oplus\mathfrak{vir}^*_\text{reg}\hspace{0.1cm},
    \label{eq.19}
\end{equation}
where both the \textit{angular supermomentum} $j$ and the \textit{supermomentum} $p$ are seen as quadratic differentials on the circle. 

We use the following natural pairing between the extended algebra and its dual: 
\begin{equation}
    \begin{split}
        \langle(j,c_1;\hspace{0.1cm}p,c_2),(X,a;\hspace{0.1cm}\alpha,b)\rangle_{\mathfrak{bms}_3}&=\int_{\text{S}^1}dx\big(j(x)X(x)+p(x)\alpha(x)\big)+ac_1+bc_2\hspace{0.1cm}.
        \label{eq.20}
    \end{split}
\end{equation}

\begin{proposition}[\cite{BarObl15}]\label{prop:proposition1.2}
\begin{equation}
        \begin{split}
        &\mathrm{ad}^*_{\mathfrak{bms}_{3_\mathrm{reg}}}(X;\alpha)\cdot(j,c_1;\hspace{0.1cm}p,c_2)=\bigg(\delta_{(X,\alpha)}j\cdot dx^2,0;\hspace{0.1cm}\delta_{(X,\alpha)}p\cdot dx^2,0\bigg)\hspace{0.1cm},\\
            \mathrm{with}\hspace{0.1cm} &\delta_{(X,\alpha)}j=Xj^{(1)}+2jX^{(1)}-c_1X^{(3)}+ \alpha p^{(1)}+2p\alpha^{(1)}-c_2\alpha^{(3)}\hspace{0.1cm},\\
            \mathrm{and}\hspace{0.1cm} & \delta_{(X,\alpha)}p=Xp^{(1)}+2pX^{(1)}-c_2X^{(3)}\hspace{0.1cm}.
        \end{split}
        \label{eq.21}
    \end{equation}
\end{proposition}

It is well-known that the coadjoint action on $\mathfrak{vir}$ corresponds to the infinitesimal transformation of the stress-energy tensor in a $2\mathrm{d}$-CFT under a holomorphic change of coordinates (see for instance \cite{DiFMatSén97}). Equivalently (\cite{Bak88,Gie91}), we have that the Schrödinger operator $\partial_x^2+p(x)$ is \textit{conformally covariant} under $\mathrm{Diff}(S^1)$.

In $(2+1)$-Einstein gravity with negative cosmological constant, \cite{BroHen86} proved that asymptotically AdS spacetimes admit a symmetry group larger than $\mathrm{SO}(2,2)$ : the $2\mathrm{d}$-\textit{conformal group}. On the other hand, the Einstein-Hilbert action is described by a conformally invariant Liouville theory \cite{CouHenDri95}. 

For a vanishing cosmological constant, based on the Liouville action at $\mathscr{S}^+$, two limits may be taken, both leading to $\mathrm{BMS}_3$-invariant field theories: \textit{i)} a first one which has no central extension; \textit{ii)} a second one based on a Hamiltonian framework where the gravitational constant remains finite. For the latter, we speak of \textit{flat Liouville theory} (\cite{BarGomGon13}). It admits a central extension. The conserved charges in this \textit{flat} case identify with the gravitational surface charges of an asymptotically flat spacetime. Moreover, the stress tensor transforms infinitesimally according to \textbf{Proposition~\ref{prop:proposition1.2}} for $c_1=0$ and $c_2=3/G$ \cite{BarGomGon13}. Hence, the supermomentum $p$ and the angular supermomentum $j$ relate to the components of this stress tensor. 

\section{Integrability in the bms$_3$ scenario}\label{sec:section2}

We properly assess the existence of an integrable system built on the $\mathfrak{bms}_3$ algebra. Specifically, we use a characterization due to \cite{GelDor81} to construct a Hamiltonian pair on the centrally extended version of this algebra. We prove that the resulting hierarchy is indeed integrable. An alternative $\tau$-scheme description is also presented. 

Before delving into the details, we believe that a short overview of \textit{infinite-dimensional integrability} through the eye of the \textit{Lenard-Magri procedure} may be profitable. In particular, after saying a few words about the \textit{variational complex}, we emphasize the fundamental role played by \textit{bi-Hamiltonian pairs} and \textit{Nijenhuis operators.} 

\subsection{The background: Integrability through bi-hamiltonian structures on the variational complex}\label{sec:section2.1}

\subsection*{A few words on the variational complex}

Let $\mathfrak{a}$ be a Lie algebra and $\big(\Omega,\mathbf{d}\big)$ be a graded complex of linear spaces $\Omega=\bigoplus\Omega^q$ whose universal differential mapping $\mathbf{d}$ satisfies $\mathbf{d}\Omega^q\subset \Omega^{q+1}$ and $\mathbf{d}^2\equiv 0$. We speak of an $\mathfrak{a}$-\textit{complex} $\big(\Omega,\mathbf{d},\mathfrak{a}\big)$ whenever there exists an interior mapping for each $a\in\mathfrak{a}$. For any linear subspace $\mathrm{Z}\subset\mathfrak{a}$, we denote by $\mathfrak{a}_\mathrm{Z}$ the corresponding centralizer. 

A $q$-form $\omega\in\Omega^q$ is said to be \textit{equivalent to} $0$, denoted by $\omega\sim 0$, if there exists some $\omega_k\in\Omega^q$ and $z_k\in\mathrm{Z}$ such that $\omega=\sum_k\Lie_{z_k}\omega_k$. Similarly, we say $\omega_1\sim\omega_2$ if, and only if, $\omega_1-\omega_2\sim 0$. Let $\Omega^q_\mathrm{Z}=\big\{\omega\in\Omega^q | \omega \sim0\big\}$ and set $\widetilde{\Omega}^q:=\Omega^q/\Omega^q_\mathrm{Z}$. Since the differential and the interior mappings remain naturally defined in $\widetilde{\Omega}$, we call $\big(\widetilde{\Omega},\mathbf{d},\mathfrak{a}_\mathrm{Z}\big)$ the \textit{factor-complex} of $\big(\Omega,\mathbf{d},\mathfrak{a}\big)$.

Take $\mathscr{A}$ the ring of polynomials in some symbols $u_{\alpha}^{(i)}$, $i\in\mathbb{N}$ and $\alpha$ being an index running over a countable set $\mathcal{I}$. The set of derivations 
\begin{equation}
    \partial_{f}=\sum_{\alpha, i}f_{\alpha,i}\frac{\partial}{\partial u_\alpha^{(i)}}, \hspace{0.2cm} f_{\alpha, i}\in\mathscr{A}\hspace{0.1cm},
    \label{eq.22}
\end{equation}
is a natural Lie algebra $\mathrm{Der}(\mathscr{A})$ over this ring with respect to the bracket:
\begin{equation}
    \big[\partial_f, \partial_g\big]_{\mathrm{Der}\mathscr{A}}=\sum_{\alpha, \beta, i,j}\bigg(f_{\beta,j}\frac{\partial g_{\alpha, i}}{\partial u_\beta^{(j)}}-g_{\beta, j}\frac{\partial f_{\alpha, i}}{\partial u_\beta^{(j)}}\bigg)\frac{\partial}{\partial u_\alpha^{(i)}} \hspace{0.1cm} .
    \label{eq.23}
\end{equation}

Let $\mathrm{Z}$ be the $1$-dim subspace of $\mathrm{Der}(\mathscr{A})$ whose derivations are proportional to the total derivative $\partial$:
\begin{equation}
    \partial=\sum_{\alpha, i}u_\alpha^{(i+1)}\frac{\partial}{\partial u_\alpha^{(i)}}\hspace{0.1cm}. 
    \label{Eq.24}
\end{equation}

It can be check that 
\begin{equation}
    \mathrm{Der}(\mathscr{A})_\mathrm{Z}=\bigg\{\partial_{f}=\sum_{\alpha,i}f_\alpha^{(i)}\frac{\partial}{\partial u_\alpha^{(i)}}, \hspace{0.2cm} f_\alpha^{(i)}=\partial^if_\alpha\bigg\}\hspace{0.1cm},
    \label{Eq.25}
\end{equation}
so that the above centralizer is in one-to-one correspondence with the set of sequences $\bar{\mathscr{A}}:=\big\{\bar{f}=\{f_\alpha\},\alpha\in\mathcal{I}\big\}$ \cite{GelDor79}. We then confuse any $\bar{f}\in\bar{\mathscr{A}}$ with its derivation $\partial_{\bar{f}}$ and we call such derivations \textit{(evolutionary) vector fields}. From the Lie bracket on $\mathrm{Der}(\mathscr{A})$ one deduces a Lie bracket on $\bar{\mathscr{A}}$, making it a well-defined Lie algebra. 

In the following, we will use the column-vector index notation $\bar{f}\equiv(f_\alpha)^\top$ for vector fields, equally with the abused notation $f\equiv\sum_{\alpha\in\mathcal{I}}f_\alpha\partial/\partial u_\alpha$, being understood that the latter is identified with $\sum_{\alpha,i}f_\alpha^{(i)}\partial/\partial u_\alpha^{(i)}$.

Consider the \textit{de Rham complex} $\Omega(\mathscr{A})$ of $\mathscr{A}$ whose linear spaces are given by:
\begin{equation}
    \Omega^q(\mathrm{Der}\mathscr{A})=\big\{\omega=\sum\omega^{i_1,\cdots,i_q}_{\alpha_1,\cdots,\alpha_q}\mathbf{d}u_{\alpha_1}^{(i_1)}\wedge\cdots\wedge\mathbf{d}u_{\alpha_q}^{(i_q)}, \hspace{0.1cm} \omega_{\alpha_1,\cdots,\alpha_q}^{i_1,\cdots,i_q}\in\mathscr{A}\big\}\hspace{0.1cm}
    \label{eq.26}
\end{equation}

\begin{definition}\label{def:definition2.1}
    The \textit{variational complex} is the factor-complex $\big(\widetilde{\Omega}(\mathscr{A}),\mathbf{d}, \bar{\mathscr{A}}\big)$ connected with $\mathrm{Z}$ ($=\mathbb{R}\partial$).
\end{definition}

The formal integral notation $\int\cdot$ associates with $\omega\in\Omega$ its equivalence class $\tilde{\omega}\in\widetilde{\Omega}(\mathscr{A})$, so that $\tilde{\omega}=\int\omega$.

$0$-forms on the variational complex consist of functionals $\widetilde{\mathcal{F}}\in\widetilde{\mathscr{A}}\equiv\widetilde{\Omega}^0=\mathscr{A}/\partial\mathscr{A}$. Given some $f\in\mathscr{A}$, its functional is $\widetilde{\mathcal{F}}=\int f$. Regarding the equivalence relation previously introduced, we would then identify two functionals whenever they differ by a total derivative. 

It is convenient to introduce the \textit{variational derivative} $\delta_{u_\alpha}\equiv\frac{\delta}{\delta u_\alpha}:\mathscr{A}\longrightarrow\mathscr{A}$:
\begin{equation}
    \frac{\delta}{\delta u_\alpha}=\sum_i\big(-\partial\big)^i\frac{\partial}{\partial u_\alpha^{(i)}}\hspace{0.1cm} .
    \label{eq.27}
\end{equation}

The latter can actually be extended as a linear map on $\widetilde{\mathscr{A}}$ (essentially because its composition with $\partial$ is trivial) and exact $1$-forms are given by: 
\begin{equation}
    \mathbf{d}\widetilde{\mathcal{F}}=\sum
    _{\alpha}\frac{\delta\widetilde{\mathcal{F}}}{\delta u_\alpha}\mathbf{d}u_\alpha\hspace{0.1cm} \in \widetilde{\Omega}^1\hspace{0.1cm}.
    \label{eq.28}
\end{equation}

More generally, and similarly to vector fields, arbitrary $1$-forms $\tilde{\omega}_\xi\in\widetilde{\Omega}^1$ are characterized by sequences $\xi=\{\xi_\alpha\}_{\alpha\in\mathcal{I}}$ such that $\tilde{\omega}_\xi=\sum_{\alpha}\xi_\alpha\mathbf{d}u_\alpha$. Crucially, for all $\xi\in\widetilde{\Omega}^1$ and $\bar{f}\in\bar{\mathscr{A}}$, there exists a non-degenerate pairing \cite{Dor93}:
\begin{equation}
    (\xi,\bar{f}):=\xi(\bar{f)}=\int\sum_\alpha\xi_\alpha f_\alpha\hspace{0.5cm}\big(=\mathbf{i}_{\bar{f}}\xi\big)\hspace{0.1cm}.
    \label{eq.29}
\end{equation} 

Finally, we mention that $\widetilde{\mathscr{A}}$ is a left $\bar{\mathscr{A}}-$module with respect to the action: 
\begin{equation}
    \bar{f}\cdot\widetilde{\mathcal{F}}=\int\sum_\alpha f_\alpha\frac{\delta\widetilde{\mathcal{F}}}{\delta u_\alpha} \hspace{0.3cm} \big(=\mathbf{i}_{\bar{f}}\hspace{0.1cm}\mathbf{d}\widetilde{\mathcal{F}}\big)\hspace{0.1cm}.
    \label{eq.30}
\end{equation}

\subsection*{Integrability through the Lenard-Magri scheme}

In order to cope with the formal definition of integrability presented in the \hyperref[sec:introduction]{\textbf{Introduction}}, we shall deal with \textit{Hamiltonian flows} or, in other words, non-linear PDEs characterized by a so-called \textit{Hamiltonian operator} $\mathscr{E}:\widetilde{\Omega}^1\longrightarrow\bar{\mathscr{A}}$. By the latter, we understand a linear operator that respects both:
\begin{equation}
        \textit{i)} \hspace{0.5cm} (\mathscr{E}\xi,\eta)+(\xi,\mathscr{E}\eta)=0,\hspace{0.5cm} \xi,\eta\in\widetilde{\Omega}^1\hspace{0.1cm}.
        \label{eq.31}
\end{equation}
\begin{equation}
        \textit{ii)}\hspace{0.5cm} \big(\Lie_{\mathscr{E}\xi}\eta,\mathscr{E}\pi\big)+\circlearrowleft(\xi,\eta,\pi)=0,\hspace{0.5cm} \xi,\eta,\pi\in\widetilde{\Omega}^1\hspace{0.1cm},
        \label{eq.32}
\end{equation}
being understood that $\circlearrowleft(\xi,\eta,\pi)$ denotes the cyclic permutation of the $1$-forms in parenthesis.

We have omitted the overline on $\mathscr{E}\xi$ since from the meaning of $\mathscr{E}$ the latter is clearly a vector field. We will stick with this convention. 

To such a Hamiltonian operator one associates a Lie bracket $\{.,.\}_\mathscr{E}:\widetilde{\mathscr{A}}\times\widetilde{\mathscr{A}}\rightarrow\widetilde{\mathscr{A}}$:
\begin{equation}
    \big\{\widetilde{\mathcal{F}},\widetilde{\mathcal{G}}\big\}_\mathscr{E}:=\int\sum_\alpha\bigg(\mathscr{E}\frac{\delta\widetilde{\mathscr{F}}}{\delta u}\bigg)_\alpha\frac{\delta\widetilde{\mathscr{G}}}{\delta u_\alpha}\equiv\big(\mathscr{E}\mathbf{d}\widetilde{\mathcal{F}},\mathbf{d}\widetilde{\mathcal{G}}\big)\hspace{0.1cm},
    \label{eq.33}
\end{equation}
together with a \textit{Hamiltonian vector field} for some functional $\widetilde{\mathcal{F}}$:
\begin{equation}
    \partial_{\mathscr{E}\mathbf{d}\widetilde{\mathcal{F}}}\hspace{0.1cm}\cdot=\mathscr{E}\mathbf{d}\widetilde{\mathcal{F}}\hspace{0.1cm}\cdot\equiv\big\{\widetilde{\mathcal{F}},\cdot\big\}_\mathscr{E}\hspace{0.1cm}.
    \label{eq.34}
\end{equation}

Given some $\xi_{n-1}\in\widetilde{\Omega}^1$, $n\geq 0$, we call the vector field $\overline{\mathcal{K}}_n:=\mathscr{E}\xi_{n-1}$ associated with $\mathscr{E}$ a \textit{(Hamiltonian) flow}. It comes with a \textit{(Hamiltonian) flow equation}:
\begin{equation}
    \frac{\partial}{\partial t}\cdot\equiv \overline{\mathcal{K}}_n[\cdot] =\mathscr{E}\xi_{n-1}[\cdot]\hspace{0.1cm},
    \label{eq.35}
\end{equation}
valid for any function(al)s. 

For $\mathscr{E},\mathscr{D}:\widetilde{\Omega}^1\longrightarrow\bar{\mathscr{A}}$ two Hamiltonian operators, we say that $(\mathscr{E},\mathscr{D})$ is a \textit{Hamiltonian pair} (equivalently, that the are \textit{compatible}, or that they are a \textit{bi-Hamiltonian} structure) if any linear combination of the two is still a Hamiltonian operator. 

With such a bi-hamiltonian pair, assuming a flow $\overline{\mathcal{K}}_n$ lies in the image of $\mathscr{D}$, it becomes possible to write it in two distinct ways: $\overline{\mathcal{K}}_n=\mathscr{E}\xi_{n-1}=\mathscr{D}\xi_n$, for some $\xi_n\in\widetilde{\Omega}^1$.

Provided that $\mathscr{D}$ is non-degenerate (so that $\mathscr{D}^{-1}$ makes sense), the operator $\mathscr{R}:=\mathscr{E}\circ\mathscr{D}^{-1}:\bar{\mathscr{A}}\to\bar{\mathscr{A}}$ leads to a new flow $\overline{\mathcal{K}}_{n+1}:=\mathscr{R}\overline{\mathcal{K}}_n=\mathscr{E}\xi_n$. If this flow lies again in the image of $\mathscr{D}$, a recursive pattern is suggested: $\overline{\mathcal{K}}_{n+1}=\mathscr{E}\xi_n=\mathscr{D}\xi_{n+1}$, $\xi_{n+1}$ being another $1$-form. $\mathscr{R}$ is then said to be a \textit{Nijenhuis operator}, in the sense that its associated Nijenhuis torsion vanishes: 
\begin{equation}
    \big[\mathscr{R}\bar{f},\mathscr{R}\bar{g}\big]_{\bar{\mathscr{A}}}-\mathscr{R}\big[\mathscr{R}\bar{f},\bar{g}\big]_{\bar{\mathscr{A}}}-\mathscr{R}\big[\bar{f},\mathscr{R}\bar{g}\big]_{\bar{\mathscr{A}}}+\mathscr{R}^2\big[\bar{f},\bar{g}\big]_{\bar{\mathscr{A}}}=0\hspace{0.1cm},\hspace{0.5cm}\forall\bar{f},\bar{g}\in\bar{\mathscr{A}}\hspace{0.1cm}.
    \label{eq.36}
\end{equation}

In practice, the procedure is initiated by looking at two functionals $\widetilde{\mathcal{H}}_{-1}$ and $\widetilde{\mathcal{H}}_0$ to get a first \textit{bi-hamiltonian flow} $\overline{\mathcal{K}}_0=\mathscr{E}\mathbf{d}\widetilde{\mathcal{H}}_{-1}=\mathscr{D}\mathbf{d}\widetilde{\mathcal{H}}_0$. 

If one shows that there are actually infinitely many flows $\overline{\mathcal{K}}_n$ together with infinitely many $1$-forms $\xi_n$ such that the so-called \textit{Lenard-Magri scheme} holds: 
\begin{equation}
    \frac{\partial}{\partial t_n}\cdot\equiv \overline{\mathcal{K}}_n[\cdot]=\mathscr{E}\xi_{n-1}[\cdot]=\mathscr{D}\xi_n[\cdot]\hspace{0.5cm} n\geq0\hspace{0.1cm},
    \label{eq.37}
\end{equation}

\vspace{1em}
\begin{center}
\begin{tikzpicture}[>=Stealth]

\node (dH_-1) at (0,0) {$\xi_{-1}$};
\node (K_0)  at (2,0.8) {$\overline{\mathcal{K}}_0$};
\node (dH_0) at (4,0) {$\xi_0$};
\node (K_1)  at (6,0.8) {$\overline{\mathcal{K}}_1$};
\node (dH_1) at (8,0) {$\xi_1$};
\node (K_2)  at (10,0.8) {$\overline{\mathcal{K}}_2$};
\node (dH_2) at (12,0) {$\xi_2$};

\draw[->] (dH_-1) -- node[midway, sloped, above] {$\mathscr{E}$} (K_0);

\draw[<-] (K_0) -- node[midway, sloped, below] {$\mathscr{D}$} (dH_0);

\draw[->] (dH_0) -- node[midway, sloped, above] {$\mathscr{E}$} (K_1);
\draw[<-] (K_1) -- node[midway, sloped, below] {$\mathscr{D}$} (dH_1);

\draw[->] (dH_1) -- node[midway, sloped, above] {$\mathscr{E}$} (K_2);
\draw[<-] (K_2) -- node[midway, sloped, below] {$\mathscr{D}$} (dH_2);

\draw[->] (dH_2) -- ++(0.75,0.4);

\node[rotate=27] at (13.3,0.7) {$\cdots$};

\end{tikzpicture}
\end{center}
\vspace{1em} 

then, all the $1$-forms are exact \cite{Dor93}: $\xi_n=\mathbf{d}\widetilde{\mathcal{H}}_n$, and the associated functionals are called the \textit{Hamiltonians}. Thus, whenever $\forall n,m\geq0$
\begin{equation}
    \boxed{
    \begin{split}
    \textit{i)}&\hspace{1cm} \big\{\widetilde{\mathcal{H}}_n,\widetilde{\mathcal{H}}_m\big\}_\mathscr{H}=\big\{\widetilde{\mathcal{H}}_n,\widetilde{\mathcal{H}}_m\big\}_\mathscr{D}=0\hspace{0.1cm}\text{with} \\
    \textit{ii)}&\hspace{1cm} \big[\overline{\mathcal{K}}_n,\overline{\mathcal{K}}_m\big]_{\bar{\mathscr{A}}}=0\hspace{0.1cm},\\
    &\textbf{The resulting hierarchy of flows $\partial/\partial t_n\equiv\overline{\mathcal{K}}_n$ is said to be integrable.}
    \end{split}
    }
\label{eq.38}
\end{equation}

\textit{A word on terminology.} We may write that a given `\textit{bi-Hamiltonian pair is integrable}'. It is understood that this misuse of terminology should be understood in light of \icref{eq.37,eq.38}. Note that \textit{ii)} is a consequence of \textit{i)}.

\subsection{A bms$_3$ bi-Hamiltonian structure}\label{sec:section2.2}

Consider a linear operator $\mathscr{E}:\widetilde{\Omega}^1\longmapsto\bar{\mathscr{A}}$ as the differential matrix operator whose entries are of the form
\begin{equation}
    \mathscr{E}_{\alpha\beta}=\sum_{\mu,k,l}a^\mu_{\alpha\beta kl}u_\mu^{(k)}\partial^{l}, \hspace{0.5cm} a^\mu_{\alpha\beta kl}\in\mathbb{R}\hspace{0.1cm}.
    \label{eq.39}
\end{equation}

As we have seen, whenever it is \textit{Hamiltonian} there exists on $\widetilde{\mathscr{A}}$ a well-defined \textit{Lie bracket} attached with it : 
\begin{equation}
    \big\{\widetilde{\mathcal{F}},\widetilde{\mathcal{G}}\big\}_\mathscr{E}=\int \sum_{\alpha,\beta}\bigg(\mathscr{E}_{\alpha\beta}\frac{\delta\widetilde{\mathcal{F}}}{\delta u_\beta}\bigg)\cdot\frac{\delta\widetilde{\mathcal{G}}}{\delta u_\alpha}\equiv \mathscr{E}\mathbf{d}\widetilde{\mathcal{F}}\cdot\widetilde{\mathcal{G}}\hspace{0.1cm}.
    \label{eq.40}
\end{equation}

For such operators, their Hamiltonian character strongly relies on the existence of a Lie algebra structure on $\widetilde{\Omega}^1$. To understand it, it is first convenient to introduce a relabeling of the previous constants 
\begin{equation}
    c^\mu_{\alpha\beta kl}=\sum_qa^\mu_{\beta\alpha,l+q,k-q}\big(-1\big)^{l+q}\binom{l+q}{q}\hspace{0.1cm}.
    \label{eq.41}
\end{equation}

In the sequel, we will base our construction of a $\mathfrak{bms}_3$ bi-Hamiltonian hierarchy on the following result, originally stated in \cite{GelDor81}:

\begin{theorem}\label{thm:theorem2.2}
    The operator $\mathscr{E}$ is Hamiltonian if, and only if, $\widetilde{\Omega}^1$ is  a Lie algebra with respect to the bilinear operation
    \begin{equation}
        \big[\xi,\eta\big]_{\widetilde{\Omega}^1}=\zeta \hspace{0.3cm} \mathrm{where} \hspace{0.3cm} \zeta_\mu=\sum_{\alpha,\beta,k,l} c^\mu_{\alpha\beta kl}\xi_\alpha^{(k)}\eta_\beta^{(l)}\hspace{0.1cm}.
        \label{eq.42}
    \end{equation}
\end{theorem}

The Hamiltonian nature of an operator is then captured by the existence of an infinite-dimensional Kirillov-Konstant structure (we shall recover this idea in a geometrical framework in \bcref{sec:section4}). This idea is translated in the formal notation\footnote{We chose the Lie bracket in $\mathrm{Vect}(S^1)$ to be the opposite of the usual Lie bracket of vector fields; this sign subtlety is taken into account in the bracket above.}
\begin{equation}
    \big(\mathscr{E}\xi,\eta\big)=-\big(u,\big[\xi,\eta\big]_{\widetilde{\Omega}^1}\big)\hspace{0.1cm}.
    \label{eq.43}
\end{equation}

Therefore, the non-centrally extended $\mathfrak{bms}_3$ algebra is attached to a Hamiltonian structure. We introduce on $\widetilde{\Omega}^1$ the bilinear operation: 
\begin{equation}
    \big[(\xi_1,\xi_2),(\eta_1,\eta_2) \big]_{\widetilde{\Omega}^1}:=\big([\xi_1,\eta_1]_{\mathrm{Vect}(S^1)}, [\xi_1,\eta_2]_{\mathrm{Vect}(S^1)}-[\eta_1,\xi_2]_{\mathrm{Vect}(S^1)}\big)\hspace{0.1cm}.
    \label{eq.44}
\end{equation}
It is then straightforward to deduce that a first Hamiltonian operator is given by
\begin{equation}
    \begin{pmatrix}
        2u_1\partial+u_1^{(1)} & 2u_2\partial+u_2^{(1)} \\ 2u_2\partial+u_2^{(1)} & 0
    \end{pmatrix}\hspace{0.1cm}.
    \label{eq.45}
\end{equation}

The previous statement and result can also be found at the level of the coefficients $c^\alpha_{\alpha\beta kl}$. Consider a set of elements $x_{\alpha l}$, $\alpha\in\mathcal{I}$ a countable set and $l\in\mathbb{Z}$, together with $U$ the set of finite linear combinations of such $x_{\alpha l}$ with real coefficients. If we denote by $\varphi^\mu_{\alpha\beta}(l,m)$ a polynomial in both $l$ and $m$ that depends on the three indices $\alpha,\beta,\mu\in\mathcal{I}$, it is common to view them as \textit{structure functions} on $U$ with respect to the bilinear operation
\begin{equation}
    \big[x_{\alpha l},x_{\beta m}]=\sum_\mu\varphi^\mu_{\alpha\beta}(l,m)x_{\mu,l+m}\hspace{0.1cm}.
    \label{eq.46}
\end{equation}

For $\varphi_{\alpha\beta}^\mu(l,m)=\sum_{q,r}c^\mu_{\alpha\beta qr}l^qm^r$, the equivalence of \textbf{Theorem~\ref{thm:theorem2.2}} is recovered at the level of structure functions whenever they define a Lie algebra on $U$ \cite{GelDor81}. 

Regarding the commutation relations \icref{eq.11}, valid for an asymptotically flat spacetime, we get back to the previous Hamiltonian operator from (recall that $x_1\equiv X$ and $x_2\equiv\alpha$)
\begin{equation}
    \begin{split}
    &\varphi^1_{11}(l,m)=\varphi_{12}^2(l,m)=l-m\hspace{0.1cm},\\
    &\varphi_{11}^2(l,m)=\varphi_{12}^1(l,m)=\varphi_{22}^1(l,m)=\varphi_{22}^2(l,m)=0\hspace{0.1cm}.
    \end{split}
    \label{eq.47}
\end{equation}

On top of that, the possibility of extending centrally a Lie algebra is deeply connected to the existence of an extra compatible Hamiltonian structure. Consider another differential matrix operator $\mathscr{D}$ whose entries are constants:
\begin{equation}
    \mathscr{D}_{\alpha\beta}=\sum_l d_{\alpha\beta l}\partial^l, \hspace{0.5cm} d_{\alpha\beta l}\in\mathbb{R}\hspace{0.1cm}.
    \label{eq.48}
\end{equation}

It follows immediately that such an operator is Hamiltonian; but more importantly: it is compatible with $\mathscr{E}$ if and only if $\mathscr{E}+\mathscr{D}$ is again Hamiltonian. It was proved in \cite{GelDor81} that the latter is in one-to-one correspondence with the existence of a skew-symmetric bilinear form $\langle\xi,\eta\rangle:=-\big(\mathscr{D}\xi,\eta\big)$ that is a $2$-cocycle.

In our case, such a two-parametric cocycle naturally stems from the centrally extended $\mathfrak{bms}_3$ algebra: 
\begin{equation}
    \begin{split}
    \langle(\xi_1,\xi_2),(\eta_1,\eta_2)\rangle_{\kappa_1,\kappa_2}:=&\int\bigg(c_1\xi_1^{(3)}-\kappa_1\xi_1^{(1)}\bigg)\eta_1 +\int\bigg(c_2\xi_1^{(3)}-\kappa_2\xi_1^{(1)}\bigg)\eta_2\\
    &-\int\bigg(c_2\eta_1^{(3)}-\kappa_2\eta_1^{(1)}\bigg)\xi_2\hspace{0.1cm},\hspace{0.5cm} \kappa_1,\kappa_2\in\mathbb{R}\hspace{0.1cm}.
    \end{split}
    \label{eq.49}
\end{equation}

We lay stress on the fact that the constants $\kappa_1,\kappa_2$ are related to a coboundary for the centrally extended $\mathfrak{bms}_3$. It should then be remembered that $\kappa_i$ is non-zero whenever the central charge $c_i$ is. In particular, for the asymptotically flat situation, where $c_1=0$, there is no corresponding $\kappa_1$. 

As a consequence, the following operator is Hamiltonian 
\begin{equation}
    \begin{pmatrix}
        -c_1\partial^3+\kappa_1\partial+2u_1\partial+u_1^{(1)} & -c_2\partial^3+\kappa_2\partial+2u_2\partial+u_2^{(1)} \\
        -c_2\partial^3+\kappa_2\partial+2u_2\partial+u_2^{(1)} & 0
    \end{pmatrix}\hspace{0.1cm}.
    \label{eq.50}
\end{equation}

At the level of the structure functions, if we define the polynomials $d_{\alpha\beta}(l):=\sum_{k}d_{\alpha\beta k}l^k$, the compatibility condition is equivalent to the existence of a $2$-cocycle $\langle x_{\alpha l},x_{\beta m}\rangle:=d_{\alpha\beta}(m)\delta_{l+m,0}$ on $U$.

Since for the extended $\mathfrak{bms}_3$ we have that %(recall that $x_1\equiv X$ and $x_2\equiv\alpha$)
\begin{equation}
    \langle x_{1l},x_{1m}\rangle=\big(c_1l^3-\kappa_1l\big)\delta_{l+m,0}\hspace{0.1cm}, \hspace{0.2cm} \langle x_{1l},x_{2l}\rangle=\big(c_2l^3-\kappa_2l\big)\delta_{l+m,0}\hspace{0.1cm},\hspace{0.2cm} \langle x_{2l},x_{2m}\rangle=0\hspace{0.1cm},
    \label{eq.51}
\end{equation}
with again $\kappa_1,\kappa_2$ constants, we recover a second constant compatible operator so that the above \icref{eq.50} is indeed Hamiltonian. 

The next statement summarizes the results: 

\begin{proposition}\label{prop:proposition2.3}
    A $\widehat{\mathfrak{bms}}_3$ bi-Hamiltonian structure is given by the pair
    \begin{equation}
        \mathscr{E}=\begin{pmatrix}
            -c_1\partial^3+2u_1\partial+u_1^{(1)} & -c_2\partial^3+2u_2\partial+u_2^{(1)} \\ -c_2\partial^3+2u_2\partial+u_2^{(1)} & 0
        \end{pmatrix} \hspace{0.1cm} \mathrm{and} \hspace{0.1cm} \mathscr{D}=\begin{pmatrix}
            \kappa_1\partial & \kappa_2\partial 
            \\ \kappa_2\partial & 0
        \end{pmatrix}\hspace{0.1cm}.
    \label{eq.52}
    \end{equation}
\end{proposition}

\begin{remark}\label{rem:remark2.4}
    Any translation of the variables $u_\alpha\longmapsto u_\alpha+\kappa_\alpha$ is similar to a shift of the cocycle by $\langle .,.\rangle_\kappa$, which is a coboundary of the $1$-cochain $\kappa$ given by $\kappa(x_{\alpha l})=\kappa_\alpha\delta_{l,0}$. 
\end{remark}

Next, we introduce the notation
\begin{equation}
    \mathcal{J}_\alpha=-c_\alpha\partial^3+2u_\alpha\partial+u_\alpha^{(1)}, \hspace{0.3cm}\alpha\in\{1,2\}\hspace{0.1cm}.
    \label{eq.53}
\end{equation}
Since the case $(c_2=0$, $c_1\not=0)$ is not physically relevant (on top of leading to a degenerate $\mathscr{D}$), we will assume that  $c_2\not=0$. Moreover, since $\kappa_2$ refers to a choice of representative, we fix, without loss of generality, $\kappa_2\equiv 1$.

It is instructive to look at the first flows for this $\mathfrak{bms}_3$ hierarchy (one needs to bear in mind that, so far, we did not say anything regarding its integrable character). We begin by considering the exact $1$-form stemming from a Casimir with respect to $\{.,.\}_\mathscr{D}$: $\mathbf{d}\widetilde{\mathcal{H}}_{-1}=b\mathbf{d}u_1+a\mathbf{d}u_2$, $a,b\in\mathbb{R}$. Anticipating a bit with a coming result, the functional $\widetilde{\mathcal{H}}_{-1}=\int bu_1+au_2$ will correspond to the first \textit{Hamiltonian}. The flow $\overline{\mathcal{K}}_0$ is simply (recall \icref{eq.35})
\begin{equation}
    \overline{\mathcal{K}}_0=\mathscr{E}\mathbf{d}\widetilde{\mathcal{H}}_{-1}=\begin{pmatrix}
        bu_1^{(1)}+au_2^{(1)} \\ bu_2^{(1)}
    \end{pmatrix}\hspace{0.1cm}.
    \label{eq.54}
\end{equation}

Now in the perspective of the Lenard-Magri scheme, we can actually find a second functional $\widetilde{\mathcal{H}}_0$ such that $\overline{\mathcal{K}}_0=\mathscr{D}\mathbf{d}\widetilde{\mathcal{H}}_0$. It also leads to a new flow, $\overline{\mathcal{K}}_1=\mathscr{E}\mathbf{d}\widetilde{\mathcal{H}}_0$:
\begin{equation}
    \begin{split}
    &\widetilde{\mathcal{H}}_0=\int\bigg( bu_1u_2+\frac{1}{2}\big(a-\kappa_1b\big)u_2^2\bigg)\hspace{0.1cm}, \\ 
    &\overline{\mathcal{K}}_1=\begin{pmatrix}
        -bc_1u_2^{(3)}+3b\big(u_1u_2\big)^{(1)}-bc_2u_1^{(3)}-\big(a-\kappa_1b\big)c_2u_2^{(3)}+3\big(a-\kappa_1b\big)u_2u_2^{(1)} \\
        -bc_2u_2^{(3)}+3bu_2u_2^{(1)}
    \end{pmatrix}\hspace{0.1cm}.
    \end{split}
    \label{eq.55}
\end{equation}

And the latter can itself be given as the image of $\mathscr{D}$, according to the $1$-form $\widetilde{\mathcal{H}}_1$:
\begin{equation}
    \widetilde{\mathcal{H}}_1=\int\bigg(-bc_2u_1u_2^{(2)}+\frac{3b}{2}u_1u_2^2+\frac{bc_1}{2}u_2^{{(1)}^2}+\frac{1}{2}\big(a-2\kappa_1b\big)\big(u_2^3+c_2u_2^{{(1)}^2}\big)\bigg)\hspace{.1cm},\hspace{0.1cm}\cdots
    \label{eq.56}
\end{equation}

\subsection{An exact bms$_3$ bi-Hamiltonian pair}\label{sec:section2.3}

It is common to speak about a \textit{Poisson pencil} $\mathscr{E}-\lambda\mathscr{D}$, $\lambda\in\mathbb{R}$, for a given bi-Hamiltonian structure $(\mathscr{E},\mathscr{D})$. 

Denote by $\sigma'$ the Fréchet derivative of an object $\sigma$\footnote{The \textit{Fréchet derivative} for such an object (vector field, form, operator from forms to vector fields) is the linear operator (with values in the appropriate target space): $\sigma'v=\frac{d}{d\varepsilon}|_{\varepsilon=0}\sigma(u+\varepsilon v)$ .}. The Lie derivative of a linear operator $\mathscr{E}: \widetilde{\Omega}^1\longmapsto\bar{\mathscr{A}}$ along a vector field $X$ is $\Lie_{X}\mathscr{E}=\mathscr{E}'X-X'\mathscr{E}-\mathscr{E}(X^*)'$, where $*$ refers to the adjoint\footnote{We recall that the $*-$adjointness is defined by
$\big(f \partial^k\big)^*= \big(-\partial\big)^k\circ f\hspace{0.1cm}\mathrm{mod}\partial\mathscr{A}$, for any $f\in\mathscr{A}$
.}. Along a constant vector field, the Lie derivative reduces to the Fréchet derivative of the operator. In particular, given only the Hamiltonian operator $\mathscr{E}$ as in \icref{eq.52}, one checks that for $X=(1/2)\partial/\partial u_1+(1/2)\partial/\partial u_2$
\begin{equation}
    \Lie_{X}\mathscr{E}=\begin{pmatrix}
        \partial & \partial \\ \partial & 0
    \end{pmatrix}=\mathscr{D}\hspace{0.1cm}.
    \label{eq.57}
\end{equation}

The compatibility of $\Lie_X\mathscr{E}$ is guaranteed by the Hamiltonian nature of $\mathscr{E}$; and it is clear that the former is also Hamiltonian. Following \cite{DubZha01, FalLor12}, we have that 
\begin{proposition}\label{prop:proposition2.5}
    $\big(\mathscr{E}, \mathscr{D}=\Lie_X\mathscr{E}\big)$ is an exact bi-Hamiltonian structure.
\end{proposition}

The vector field $X$ is then often called a \textit{Liouville} field and is not unique. Since it is immediate that $\mathscr{D}$ is (formally) non-degenerate, $\mathscr{R}: \bar{\mathscr{A}}\longmapsto\bar{\mathscr{A}}$, $\mathscr{R}:=\mathscr{E}\circ\mathscr{D}^{-1}$ is a well-defined operator.

If we call $Y:=\mathscr{R}X=\big(u_1+\frac{1}{2}xu_1^{(1)}\big)\frac{\partial}{\partial u_1}+\big( u_2+\frac{1}{2}xu_2^{(1)}\big)\frac{\partial}{\partial u_2}$, one shows that
\begin{proposition}\label{prop:proposition2.6}
    $Y$ is a conformal symmetry of the operator $\mathscr{R}$, i.e. $\Lie_Y\mathscr{R}=\mathscr{R}$.
\end{proposition}

Actually, this last result hints that $\mathscr{R}$ is \textit{Nijenhuis}: we will properly settle it in \textbf{Corollary~\ref{cor:corollary2.9}}. Nevertheless, provided that it is indeed Nijenhuis, based on \cite{FalLor12}, we may deduce that 
\begin{proposition}\label{prop:proposition2.7}
    The family of vector fields $\{Y_n:=\mathscr{R}^{n+1}X\}_{n\in\mathbb{N}}$ spans an upper Virasoro algebra ($n\geq0$) with vanishing central charge.
\end{proposition}

Indeed, using the Leibniz rule and the property of $\mathscr{R}$ being Nijenhuis:
\begin{equation}
    \begin{split}
        \Lie_{Y_n}Y_m&=\Lie_{\mathscr{R}^nY}\big(\mathscr{R}^m\big)Y+\mathscr{R}^m\Lie_{\mathscr{R}^nY}Y=\mathscr{R}^n\Lie_Y\big(\mathscr{R}^m\big)Y-(m\leftrightarrow n)=(m-n)Y_{m+n}\hspace{0.1cm}.
    \end{split}
    \label{eq.58}
\end{equation}

The latter suggests a connection between the $\mathfrak{bms}_3$ hierarchy and the concept of \textit{master symmetries} originally introduced by \cite{Fuc83}. Regarding KdV for instance, it was shown in \cite{ZubMag91}, how these symmetries could be used to characterize certain classes of Schrödinger operators\footnote{See also \cite{DuiGrü86} for a first discussion on the subject.} (the so-called \textit{bispectral} potentials). 

Recently, this link between mastersymmetries, Darboux transformations and the KdV hierarchy was rediscovered in the context of perturbed black hole physics: see \cite{LenSop21, LenSop23, JarLenSop24}.

In this perspective, beyond a connection with some `Schrödinger-like' operators that is suggested (and which will be the proper matter of \bcref{sec:section5}), we believe that these $\mathfrak{bms}_3$ mastersymmetries should be properly investigated in future work. 

\subsection{The integrable nature of the bms$_3$ hierarchy}\label{sec:section2.4}

We have exhibited a bi-Hamiltonian structure starting from the centrally extended $\mathfrak{bms}_3$ algebra. We now properly prove that the resulting hierarchy is integrable. Recall that, in accordance with \icref{eq.38}, we need ultimately to check that $\textit{i)}$ there exist infinitely many exact $1$-forms $\mathbf{d}\widetilde{\mathcal{H}}_n$ such that the functionals $\widetilde{\mathcal{H}}_n$ are in involution with respect to both the $\mathscr{E}$- and $\mathscr{D}$-Lie brackets (the $\widetilde{\mathcal{H}}_n$ are conserved quantities for all the $\overline{\mathcal{K}}_n$ and correspond to the \textit{Hamiltonians}); and that \textit{ii)} all flows $\overline{\mathcal{K}}_n=\mathscr{E}\mathbf{d}\widetilde{\mathcal
H}_{n-1}=\mathscr{D}\mathbf{d}\widetilde{\mathcal{H}}_{n}$ commute in $\bar{\mathscr{A}}$.

We begin with the following proposition that confirms a $\mathfrak{bms}_3$-like Lenard-Magri scheme, as described in \icref{eq.37}, does hold:

\begin{proposition}\label{prop:proposition2.8}
    Let $\overline{\mathcal{K}}_0$ be defined as in \icref{eq.54} with $\overline{\mathcal{K}}_0=\mathscr{E}\mathbf{d}\widetilde{\mathcal{H}}_{-1}=\mathscr{D}\mathbf{d}\widetilde{\mathcal{H}}_0$. For each $n\geq0$, the vector field $\overline{\mathcal{K}}_n=\mathscr{R}^n\overline{\mathcal{K}}_0$ lies in the image of $\mathscr{D}$ and hence we can recursively define $\overline{\mathcal{K}}_{n+1}=\mathscr{R}\overline{\mathcal{K}}_n$, where $\mathscr{R}=\mathscr{E}\circ\mathscr{D}^{-1}$.
\end{proposition}

We refer the reader to the attached \hyperref[sec:appendix]{\textbf{Appendix}} for a proof of this statement. A direct consequence that stems from the existence of this recursion scheme is that

\begin{corollary}\label{cor:corollary2.9}
    The Nijenhuis torsion \icref{eq.36} associated with our $\mathscr{R}$ vanishes: such $\mathscr{R}$ is thus Nijenhuis.
\end{corollary}

Additionally, we deduce the existence of an infinite number of functionals
$\widetilde{\mathcal{H}}_n$, such that for all $n\geq0$, $\overline{\mathcal{K}}_n=\{\widetilde{\mathcal{H}}_{n-1},\cdot\}_\mathscr{E}=\{\widetilde{\mathcal{H}}_n,\cdot\}_\mathscr{D}$. It is a routine computation to check that they are in involution. Point \textit{i)} is thus satisfied.

Regarding point \textit{ii)}, since $\overline{\mathcal{K}}_0$ is a symmetry of the Nijenhuis operator $\mathscr{R}$ (that is $\Lie_{\overline{\mathcal{K}}_0}\mathscr{R}=0$), so are all the $\overline{\mathcal{K}}_n$. Because also $\Lie_{\overline{\mathcal{K}}_n}\mathscr{R}=\mathscr{R}^n\Lie_{\overline{\mathcal{K}}_0}\mathscr{R}$ and $\Lie_{\overline{\mathcal{K}}_n}\overline{\mathcal{K}}_m=\big(\Lie_{\overline{\mathcal{K}}_n}\mathscr{R}^m\big)\overline{\mathcal{K}}_0-\mathscr{R}^m\Lie_{\overline{\mathcal{K}}_0}\overline{\mathcal{K}}_n$, the flows indeed commute.

\begin{theorem}\label{thm:theorem2.10}
    The $\mathfrak{bms}_3$-hierarchy of nonlinear flows $\partial\cdot/\partial t_n=\mathscr{E}\mathbf{d}\widetilde{\mathcal{H}}_{n-1}[\cdot]$$(=\mathscr{D}\mathbf{d}\widetilde{\mathcal{H}}_n[\cdot])$, $n\geq0$, stemming from the bi-Hamiltonian structure \icref{eq.52}, is integrable.    
\end{theorem}

This theorem systematizes and completes, in the present framework, the result originally stated in \cite{Fue.al18}.

\subsection{A $\tau-$scheme integrability perspective for bms$_3$ in the asymptotically flat case}\label{sec:section2.5}

So far, we have based our description of an integrable hierarchy on the so-called \textit{Lenard-Magri scheme}. The \textit{$\tau$-scheme} formalism (\cite{Dor86}) offers an alternative picture. Recall that an infinite sequence of flows is given by $\partial_{t_n}u=\overline{\mathcal{K}}_n[u]$ $\forall n\in\mathbb{N}$, where $\overline{\mathcal{K}}_n[u]$ is an infinitesimal generator of a symmetry, that is, an evolutionary vector field.

We speak of a $\tau$-scheme whenever, for some fixed \textit{seed} $\overline{\mathcal{K}}_0$, there exists an evolutionary vector field $\tau$ such that \textit{i)} elements $\overline{\mathcal{K}}_n$ are obtained via the recurrence procedure 
\begin{equation}
    \overline{\mathcal{K}}_{n+1}=\big[\tau,\overline{\mathcal{K}}_n\big]\equiv\mathrm{ad}\tau\cdot\overline{\mathcal{K}}_n\hspace{.2cm}\big(=\mathrm{ad}^n\tau\cdot\overline{\mathcal{K}}_0\big)\hspace{0.1cm},
    \label{eq.59}
\end{equation}
and \textit{ii)} all the resulting symmetries mutually commute.

Whereas in the previous discussion the Nijenhuis operator was at the core of the induction pattern, now $\mathrm{ad}\tau$ plays the role of the generator and has, \textit{a priori}, pretty much nothing in common with the former. 

A $\tau$-scheme paves the way for the construction of a Hamiltonian pair. Indeed, it can be shown that for $\mathscr{D}$ a Hamiltonian operator and $\tau$ the generator, if the Lie derivative of $\mathscr{D}$ along $\tau$ is proportional to a second Hamiltonian operator $\mathscr{E}$, then the two are compatible. 

Conversely, given a bi-Hamiltonian structure, one can extract a $\tau$-scheme:

\begin{theorem}[\cite{Dor93}]\label{thm:theorem2.11}
    Let $\mathscr{D},\mathscr{E}:\widetilde{\Omega}^1\longmapsto\bar{\mathscr{A}}$ be a Hamiltonian pair with nondegenerate $\mathscr{D}$. If $\mathcal{H}^2\big(\widetilde{\Omega},\mathbf{d}\big)$ is trivial, then there exists $\tau\in\bar{\mathscr{A}}$ such that $\Lie_\tau\mathscr{D}\propto\mathscr{E}$.
\end{theorem}
The exactness of the variational complex in all positive dimensions has been established in \cite{Dor93}.

We illustrate the construction for the asymptotically flat situation $c_1=0$ and $c_2=3/G$, being understood that it generalizes immediately for arbitrary central charges. 

\begin{proposition}\label{prop:proposition2.12}
    There exists a $\tau$-scheme related to the asymptotically flat $\mathfrak{bms}_3$ bi-Hamiltonian structure whose generator $\tau\in\bar{\mathscr{A}}$ is given by: 
    \begin{equation}
        \tau=\tau_1\frac{\partial}{\partial u_1} +\tau_2\frac{\partial}{\partial u_2}\hspace{0.1cm}, 
        \label{eq.60}
    \end{equation}
\begin{equation}
    \begin{split}
        \tau_1=\frac{d}{dx}&\frac{\delta}{\delta u_2}\int x\bigg(\frac{3}{2}bu_1u_2^2-c_2bu_1u_2^{(2)}-2ac_2u_2^{{(1)}^2}\bigg)+\hat{\tau}_1\hspace{.1cm},
    \end{split}
    \label{eq.61}
\end{equation}

\begin{equation}
    \tau_2=\frac{d}{dx}\frac{\delta}{\delta u_1}\int x\bigg(\frac{3}{2}bu_1u_2^2-c_2bu_1u_2^{(2)}\bigg)+\hat{\tau}_2\hspace{.1cm},
    \label{eq.62}
\end{equation}
for $a$, $b\in\mathbb{R}$ and where $\hat{\tau}_1$ and $\hat{\tau}_2$ are respectively 
\begin{equation}
    \begin{split}
    \hat{\tau}_1&=-2c_2u_1^{(2)}+u_2^{(1)}\partial^{-1}\circ u_1+u_1^{(1)}\partial^{-1}\circ u_2+5u_1u_2\hspace{0.1cm},\\
    \hat{\tau}_2&=-2c_2u_2^{(2)}+u_2^{(1)}\partial^{-1}\circ u_2+\frac{5}{2}u_2^2\hspace{0.1cm}. 
    \end{split}
    \label{eq.63}
\end{equation}
Then, $\Lie_{\tau}\mathscr{D}=-4\mathscr{E}$.
\end{proposition}

Notice how the semi-direct structure of $\mathfrak{bms}_3$ is suggested by the above generator.

We purposefully chose to write the components of $\tau$ as above, since then the Lie derivative of $\mathscr{D}$ along it simplifies as 
\begin{equation}
    \Lie_\tau\mathscr{D}=\Lie_{\hat{\tau}}\mathscr{D}, \hspace{0.5cm} \hat{\tau}=\hat{\tau}_1\frac{\partial}{\partial u_1}+\hat{\tau}_2\frac{\partial}{\partial u_2}\hspace{0.1cm}.
    \label{eq.64}
\end{equation}

Actually, the $\tau$-induction scheme fixes $b\in\{0,1\}$. This can be easily seen from \icref{eq.59} at $n=0$, where 
\begin{equation}
    \big[\tau,\bar{\mathcal{K}_0} \big]=\begin{pmatrix}
        b\tau_1^{(1)}+a\tau_2^{(1)}-\tau_{11}'\mathcal{K}_{0,1}-\tau_{12}'\mathcal{K}_{0,2} \\
        b\tau_2^{(1)}-\tau_{22}'\mathcal{K}_{0,2}
    \end{pmatrix}=\begin{pmatrix}
        \mathcal{J}_1\cdot\delta_{u_1}\widetilde{\mathcal{H}}_1+\mathcal{J}_2\cdot\delta_{u_2}\widetilde{\mathcal{H}}_2 \\
        \mathcal{J}_2\cdot\delta_{u_1}\widetilde{\mathcal{H}}_1
    \end{pmatrix}\hspace{0.1cm},
    \label{eq.65}
\end{equation}
holds\footnote{We denote by $\tau_{ij}'$ the entries of the $2\times2$ matrix representing the Fréchet derivative of $\tau$. Moreover, recall that $\overline{\mathcal{K}}_0=\mathcal{K}_{0,1}\partial/\partial u_1+\mathcal{K}_{0,2}\partial/\partial u_2\equiv\big(bu_1^{(1)}+au_2^{(1)}\big)\partial/\partial u_1+bu_2^{(1)}\partial/\partial u_2$.} in the aforementioned condition. 

The case $b=1$ mimics the situation tackled in \cite{Fue.al18}, while $b=0$ leads to another sequence of Hamiltonians in involution, commuting with the sequence for $b=1$. The two sequences are related to one another, and this translates at the level of the Hamiltonian vector fields.

Therefore, the Lenard-Magri procedure and the $\tau$-scheme are intrinsically related. The correspondence between the two is also recovered at the level of Hamiltonians
\begin{equation}
    \widetilde{\mathcal{H}}_{n+1}\propto\big(\tau,\mathbf{d}\widetilde{\mathcal{H}}_n\big)\hspace{0.1cm},
    \label{eq.66}
\end{equation}
where $(\cdot,\cdot)$ is the pairing between $1$-forms and evolutionary vector fields (\icref{eq.29}).

For instance, one checks that 
\begin{equation}
    \widetilde{\mathcal{H}}_0=\frac{1}{3}\big(\tau,\mathbf{d}\widetilde{\mathcal{H}}_{-1}\big), \hspace{0.5cm} \widetilde{\mathcal{H}}_1=\frac{1}{5}\big(\tau,\mathbf{d}\widetilde{\mathcal{H}}_0\big), \hspace{0.2cm} \cdots
    \label{eq.67}
\end{equation}
\newpage

\section{An integrable AdS$_3$ structure and its connection with bms$_3$}\label{sec:section3}

The AdS$_3$ spacetime is a manifold diffeomorphic to $\mathrm{S}^1\times\mathbb{R}^2$ whose isometry group is $\mathrm{O}(2,2)$. We have a diffeomorphism with respect to the quotient of this group with its stabilizer $\mathrm{O}(2,1)$. 

Prior to the BMS$_3$ study, \cite{BroHen86} looked at the transformations that preserved the fall-off conditions-- so-called \textit{Brown-Henneaux conditions}-- of an \textit{asymptotically AdS$_3$ spacetime}. They showed that the corresponding asymptotic isometry algebra was given by two copies of the Witt algebra, enhancing the original $\mathfrak{so}(2,2)$ case. 

Again, we consider the centrally extended version of this asymptotically AdS$_3$ algebra, which we shall denote $\widehat{\mathfrak{ads}}_3$.

\begin{definition}\label{def:definition3.1}
    The centrally extended Lie algebra $\widehat{\mathfrak{ads}}_3$ is the Virasoro direct sum
    \begin{equation}
        \widehat{\mathfrak{ads}}_3:=\mathfrak{vir}\oplus\mathfrak{vir}\hspace{0.1cm},
        \label{eq.68}
    \end{equation}
    whose Lie bracket derives from the one of the Virasoro algebra. Denoting by $x_m$ and $\bar{x}_m$ the generators, it is given in Fourier modes by
    \begin{equation}
        \begin{split}
            [x_m,x_n]&=(m-n)x_{m+n}+\big(cm^3-\lambda m\big)\delta_{n+m,0}\hspace{0.1cm}, \\
            [x_m,\bar{x}_n]&=0\hspace{0.1cm}, \\
            [\bar{x}_m,\bar{x}_n]&=(m-n)\bar{x}_{m+n}+\big(\bar{c}m^3-\bar{\lambda}m\big)\delta_{n+m,0}\hspace{0.1cm},
        \end{split}
        \label{eq.69}
    \end{equation}
    where $\lambda$, $\bar{\lambda}\in\mathbb{R}$ and $c=\bar{c}=\frac{3l}{2G}$ represent the Brown-Henneaux central charges.
    \end{definition}

Since any translation $\lambda\delta_{n+m,0} \longmapsto\lambda'\delta_{n+m,0}$ is cohomologous to the original central extension, except when stated otherwise, we confuse $\lambda\equiv\bar{\lambda}\equiv 1$.

Let us introduce the bilinear operation 
\begin{equation}
    \big[(\xi_1,\xi_2),(\eta_1,\eta_2)\big]_{\widetilde{\Omega}_1}:=\big([\xi_1,\eta_1]_{\mathrm{Vect}(S^1)}, [\xi_2,\eta_2]_{\mathrm{Vect}(S^1)}\big)\hspace{0.1cm},
    \label{eq.70}
\end{equation}
together with the following two-parametric cocycle
\begin{equation}
    \big<(\xi_1,\xi_2),(\eta_1,\eta_2)\big>_{}:=\int\bigg(c\xi_1^{(3)}-\xi_1^{(1)}\bigg)\eta_1+\int\bigg(\bar{c}\xi_2^{(3)}-\xi_2^{(1)}\bigg)\eta_2 \hspace{0.1cm}.
    \label{eq.71}
\end{equation}

\begin{proposition}\label{prop:proposition3.2}
    An $\widehat{\mathfrak{ads}}_3$ bi-Hamiltonian structure is given by:
    \begin{equation}
        \mathscr{E}=\begin{pmatrix}
            -c\partial^3+2u_1\partial+u_1^{(1)} & 0 \\
            0 & -\bar{c}\partial^3+2u_2\partial+u_2^{(1)}
        \end{pmatrix}\hspace{0.1cm}, \hspace{0.2cm} \mathrm{and} \hspace{0.2cm} 
        \mathscr{D}=\begin{pmatrix}
            \partial & 0 \\
            0 & \partial
        \end{pmatrix}\hspace{0.1cm}.
        \label{eq.72}
    \end{equation}
\end{proposition}

From the direct structure of $\widehat{\mathfrak{ads}}_3$, these \textit{left} and \textit{right} KdV operators should not be surprising (see \cite{PérTemTro16}). For the sake of clarity, we illustrate how this double KdV hierarchy is recovered at the level of the flows. 

As before, we look for a functional $\widetilde{\mathcal{H}}_{-1}=\int bu_1+au_2$, $a$, $b\in\mathbb{R}$ which is a Casimir with respect to $\mathscr{D}$. The first flow is 
\begin{equation}
    \overline{\mathcal{K}}_0=\begin{pmatrix}
        bu_1^{(1)} \\a u_2^{(1)}
    \end{pmatrix}\hspace{0.1cm}.
    \label{eq.73}
\end{equation}

Similarly, $\widetilde{\mathcal{H}}_0=\int\frac{b}{2}u_1^2+\frac{a}{2}u_2^2$ and $\widetilde{\mathcal{H}}_1=\int\frac{b}{2}\bigg(u_1^3+cu_1^{{(1)}^2}\bigg)+\frac{a}{2}\bigg(u_2^3+\bar{c}u_2^{{(1)}^2}\bigg)$, and
\begin{equation}
    \overline{\mathcal{K}}_1=\begin{pmatrix}
        -cbu_1^{(3)}+3bu_1u_1^{(1)} \\ -a\bar{c}u_2^{{3)}}+3au_2u_2^{(1)}
    \end{pmatrix}\hspace{0.1cm}.
    \label{eq.74}
\end{equation}

The operator $\mathscr{R}:=\mathscr{E}\circ\mathscr{D}^{-1}$ is well-defined. As for the $\mathfrak{bms}_3$ situation, we may prove that for $n\geq 0$, each $\overline{\mathcal{K}}_n=\mathscr{R}\overline{\mathcal{K}}_{n-1}$ is the characteristic of a symmetry of the $\widehat{\mathfrak{ads}}_3$ hierarchy, making $\mathscr{R}$ Nijenhuis. Therefore, we conclude that 
\begin{theorem}\label{thm:theorem3.3}
    The $\widehat{\mathfrak{ads}}_3$-hierarchy of flows attached with the bi-Hamiltonian pair \icref{eq.72} is integrable. 
\end{theorem}

It is known how the \textit{flat limit} allows us to recover the asymptotically flat $\mathfrak{bms}_3$ algebra from the above $\widehat{\mathfrak{ads}}_3$ one. Finally, we show how the latter could be used to deduce the previous $\mathfrak{bms}_3$ bi-Hamiltonian structure from the two copies of KdV proper to $\widehat{\mathfrak{ads}}_3$.

Actually, because the flat limit acts at the level of the commutation relations, with the Gel'fand-Dorfman formalism in hand, it is automatic to see how it impacts the Hamiltonian structures. Following the standard construction in \cite{BarGomGon12}, we first shift the original $\widehat{\mathfrak{ads}}_3$-generators to 
\begin{equation}
    X_{m} := x_m-\bar{x}_m \hspace{0.3cm} \mathrm{and} \hspace{0.3cm} \alpha_{m}:=l^{-1}\big(x_m+\bar{x}_m\big)\hspace{0.1cm},
    \label{eq.75}
\end{equation}
and one has the associated commutation relations: 
\begin{equation}
    \begin{split}
        [X_{m},X_{n}]&=(m-n)X_{m+n}+\big(c-\bar{c}\big)m^3\delta_{n+m,0} \hspace{0.1cm}, \\
        [X_{m},\alpha_{n}]&=(m-n)\alpha_{m+n}+\bigg(\frac{c+\bar{c}}{l}m^3-m\bigg)\delta_{n+m,0}\hspace{0.1cm}, \\
        [\alpha_{m},\alpha_{n}]&=l^{-2}(m-n)X_{m+n}+l^{-2}(c-\bar{c})m^3\delta_{n+m,0}\hspace{0.1cm}.
    \end{split}
    \label{eq.76}
\end{equation}

With respect to the relabeling of the generators, we obtain the bi-Hamiltonian pair: 
\begin{equation}
    \begin{pmatrix}
        (\bar{c}-c)\partial^3+2u_1\partial+u_1^{(1)} & -\frac{c+\bar{c}}{l}\partial^3+2u_2\partial+u_2^{(1)} \\
        -\frac{c+\bar{c}}{l}\partial^3+2u_2\partial+u_2^{(1)} & \frac{\bar{c}-c}{l^2}\partial^3+\frac{2}{l^2}u_1\partial+\frac{1}{l^2}u_1^{(1)}
    \end{pmatrix}\hspace{0.2cm} \mathrm{and} \hspace{0.2cm} \begin{pmatrix}
        0 & \partial \\
        \partial & 0
    \end{pmatrix}\hspace{0.1cm}.
    \label{eq.77}
\end{equation}

Hence, taking the flat limit $l\to+\infty$, we do recover the bi-Hamiltonian structure attached to $\mathfrak{bms}_3$ (see \icref{eq.52}) in the asymptotically flat case ($c_1=0$ and $c_2=3/G$). 

\begin{remark}\label{rem:remark3.4}
    Usually, we set $c_1:=\mathrm{lim_{l\to\infty}}\big(c-\bar{c}\big)$ and $c_2:=\mathrm{lim}_{l\to\infty} l^{-1}\big(c+\bar{c}\big)$. From the Brown-Henneaux charges, it is evident that the former vanishes and the latter is proportional to $G^{-1}$.
\end{remark}

\begin{remark}\label{rem:remark3.5}
    It is well-known that the non-centrally extended $\mathfrak{bms}_3$ is isomorphic to the $2\mathrm{d}$-Galilean conformal algebra $\mathfrak{gca}_2$ (\cite{Bag10}); the construction above is similar to the Inönü-Wigner contraction of the Poincaré algebra to the Galilei algebra. At the classical level\footnote{See \cite{Obl18} for a discussion at the quantum level.}, both are identical. Thus, $\mathfrak{gca}_2$ has a bi-Hamiltonian integrable structure. In particular, it would be relevant to address how integrability could step in the question of flat space holography \cite{Bag10}.
\end{remark}

\section{A Lie-Poisson perspective on the bms$_3$ bi-Hamiltonian hierarchy}\label{sec:section4}

In this section, we adopt a geometric perspective on bi-Hamiltonian structures. It is shown how one can extract a $\mathfrak{bms}_3$ bi-Hamiltonian pair, first stated in \textbf{Proposition~\ref{prop:proposition2.3}}, through \textit{Lie-Poisson brackets} on $\mathfrak{bms}_3^*$. Crucially, this approach hints towards the existence of extra $\mathfrak{bms}_3$-like bi-Hamiltonian pairs. 

\vspace{0.5em}

A common feature of integrable systems is that their phase spaces are described by the coadjoint orbits of some appropriate Lie group. The justification stems from the standard fact that the dual space of a Lie algebra $\mathfrak{g}$ carries a natural \textit{Poisson} structure through the \textit{Lie-Poisson} bracket: 
\begin{equation}
    \big\{\tilde{\varphi},\tilde{\psi}\big\}_{\text{LP}}(m):=\big<m,\big[\mathbf{d}_m\tilde{\varphi},\mathbf{d}_m\tilde{\psi}\big]_\mathfrak{g} \big>, \hspace{0.2cm} m\in\mathfrak{g}^*, \tilde{\varphi},\tilde{\psi}\in C^\infty(\mathfrak{g}^*)\hspace{0.1cm}.
    \label{eq.78}
\end{equation}

From \cite{Kir76, Wei83, AleMal94}, we know that the symplectic leaves of the Lie-Poisson bracket are the coadjoint orbits of $\mathfrak{g}$; and in particular, the Casimirs are coadjoint-invariant functions. Any $\widetilde{\mathcal{H}}\in C^\infty(\mathfrak{g}^*)$ defines a \textit{Hamiltonian vector field} $X_{\widetilde{\mathcal{H}}}$ on $\mathfrak{g}^*$ via the Lie derivative: $\Lie_{X_{\widetilde{\mathcal{H}}}}\tilde{\psi}=\big\{\widetilde{\mathcal{H}},\tilde{\psi}\big\}$, $\forall\tilde{\psi}\in C^\infty(\mathfrak{g}^*)$. For $\tilde{\varphi}\in C^\infty(\mathfrak{g}^*)$, the Hamiltonian equations of motion with respect to the Lie-Poisson bracket are then
\begin{equation}
    \frac{d}{dt}m(t)=\text{ad}^*_\mathfrak{g}\mathbf{d}_{m(t)}\tilde{\varphi}\cdot m(t), \hspace{0.2cm} m\in \mathfrak{g}^*\hspace{0.1cm}.
    \label{eq.79}
\end{equation}

We shall see the importance of these dynamical equations in the next section. Nevertheless, the above \icref{eq.79} is, in some sense, the `geometric twin' of the algebraic formulation \icref{eq.35} of a Hamiltonian flow equation. 

The Lie-Poisson bracket is then a natural candidate to exhibit a Hamiltonian operator for a given Lie algebra, and the centrally extended $\mathfrak{bms}_3$ is no exception. 

Consider a point $m=(u_1,c_1;u_2,c_2)\in\mathfrak{bms}_3^*$. The differential of an arbitrary smooth function $\tilde{\varphi}$ in the dual is formally identified with  $\mathbf{d}_m\tilde{\varphi}\equiv\big(\delta_{u_1}\tilde{\varphi},\delta_{c_1}\tilde{\varphi}; \delta_{u_2}\tilde{\varphi},\delta_{c_2}\tilde{\varphi}\big)$. These variational components are determined via the following identity, valid for any $(j dx^2,a; pdx^2,b)\in\mathfrak{bms}_3^*$: 
\begin{equation}
    \bigg<\big(jdx^2,a;pdx^2,b\big),\bigg(\frac{\delta\tilde{\varphi}}{\delta u_1},\frac{\delta\tilde{\varphi}}{\delta c_1}; \frac{\delta\tilde{\varphi}}{\delta u_2}, \frac{\delta\tilde{\varphi}}{\delta c_2}\bigg)\bigg>=\frac{d}{d\varepsilon}\bigg|_{\varepsilon=0}\tilde{\varphi}\big(u_1+\varepsilon j, c_1+\varepsilon a; u_2+\varepsilon p, c_2+\varepsilon b\big),
    \label{eq.80}
\end{equation}
whose pairing on the left-hand side is the one chosen to identify $\mathfrak{bms}_3$ with its dual. Then, one recognizes that the Lie-Poisson bracket of $\mathfrak{bms}_3$ coincides with the formal bracket associated with the first Hamiltonian structure identified above: 

\begin{equation}
    \big\{\tilde{\varphi},\tilde{\psi}\big\}_\mathrm{LP}^{\mathfrak{bms}_3^*}(m)=\int dx\bigg(\delta_{u_1}\tilde{\varphi}\cdot\mathcal{J}_1\delta_{u_1}\tilde{\psi}+\delta_{u_1}\tilde{\varphi}\cdot\mathcal{J}_2\delta_{u_2}\tilde{\psi}+\delta_{u_2}\tilde{\varphi}\cdot\mathcal{J}_2\delta_{u_1}\tilde{\psi}\bigg)\equiv\big\{\tilde{\varphi},\tilde{\psi}\big\}_\mathscr{E}\hspace{0.1cm}.
    \label{eq.81}
\end{equation}

This is a natural result since we have previously seen how such a structure was connected to a Kirillov-Konstant one (recall \icref{eq.43}). 

As for a second compatible Hamiltonian structure, by taking inspiration from the finite dimensional setting, it can be found from the Lie-Poisson bracket at a fixed point $m_0\in\mathfrak{bms}_3^*$ : we speak about the \textit{frozen point method}, and the bracket is denoted by $\{,\}_\mathrm{0}$:
\begin{equation}
    \big\{\tilde{\varphi},\tilde{\psi}\big\}^{\mathfrak{bms}_3^*}_\mathrm{0}(m_0):=\big<m_0,\big[\mathbf{d}_m\tilde{\varphi},\mathbf{d}_m\tilde{\psi}\big]_{\mathfrak{bms}_3}\big>\hspace{0.1cm}.
    \label{eq.82}
\end{equation}

It is immediate to check that the compatibility between the Lie-Poisson bracket and its frozen partner holds for every choice of $m_0$. We call a \textit{Lie-Poisson bi-Hamiltonian pair}, the resulting pair $\big(\{.,.\}_\mathrm{LP}^{\mathfrak{bms}_3^*},\{.,.\}_\mathrm{0}^{\mathfrak{bms}_3^*}\big).$

Regarding $\mathfrak{bms}_3$, we wonder how to find such a frozen point. We first state the result before explaining it.
\begin{proposition}\label{prop:proposition4.1}
    A Lie-Poisson bi-Hamiltonian pair for $\mathfrak{bms}_3$ is found at $m_0=\big(\kappa_1/2\cdot dx^2,0; 1/2\cdot dx^2,0\big)$:
    \begin{equation}
        \big\{\tilde{\varphi},\tilde{\psi}\big\}^{\mathfrak{bms}_3^*}_\mathrm{0}(m_0)\equiv\big\{\tilde{\varphi},\tilde{\psi}\big\}_\mathscr{D}\hspace{0.1cm}.
        \label{eq.83}
    \end{equation}
\end{proposition}

Let $(j_0\cdot dx^2,\kappa_1c_1^0; p_0\cdot 
dx^2,c_2^0)$ be a constant frozen point in $\mathfrak{bms}_3^*$ and identify  the vector fields $\big(X\partial_x; \alpha\partial_x\big)$ with the variational components with respect to $u_1$ and $u_2$ of $\mathbf{d}_m\tilde{\varphi}$ for $\tilde{\varphi}\in C^\infty(\mathfrak{bms}_3^*)$. The local coadjoint action writes
\begin{equation}
    \mathrm{ad}^*_{\mathfrak{bms}_3}\big(X;\alpha\big)\cdot\big(j_0,\kappa_1c_1^0; p_0,c_2^0\big)=\mathcal{\mathcal{I}}\cdot\mathscr{D}\begin{pmatrix}
        X \\ \alpha
    \end{pmatrix} \hspace{0.2cm} \mathrm{for} \hspace{0.2cm} \mathcal{I}=\begin{pmatrix}
        \theta_2-\iota_2\partial^2 & \theta_1-\iota_1\partial^2 \\
        0 & \theta_2-\iota_2\partial^2
    \end{pmatrix}, 
    \label{eq.84}
\end{equation}
where we have defined $\theta_1=2(j_0-\kappa_1p_0)$, $\theta_2=2p_0$, $\iota_1=\kappa_1(c_1^0-c_2^0)$, $\iota_2=c_2^0$.

We view $\mathcal{I}:\mathfrak{bms}_3\to\mathfrak{bms}_3^*$ as an inertia operator in the terminology developed by \cite{Arn89}. More importantly, such an operator may be directly related to the pairing between $\mathfrak{bms}_3$ and its dual, via the so-called $H^1$-energies (a $4$-parameter family of quadratic forms on $\mathfrak{bms}_3$): 

\begin{equation}
    \begin{split}
    &\big<(j\partial_x,c_1; p\partial_x,c_2),(X\partial_x,a;\alpha\partial_x,b)\big>_{\mathrm{H}^1_{(\theta_1,\iota_1;\theta_2,\iota_2)}}\\
    &=\int_{\mathrm{S}^1}dx\bigg(\theta_1j\alpha+\theta_2\big(jX+p\alpha\big)+\iota_1j^{(1)}\alpha^{(1)}+\iota_2\big(j^{(1)}X^{(1)}+p^{(1)}\alpha^{(1)}\big)\bigg)+c_1a+c_2b\\
    &=\big<\big(\mathcal{I}(j,c_1; p,c_2)\cdot dx^2\big),  \big(X\partial_x, a; \alpha\partial_x,b\big) \big>_{\mathfrak{bms}_3}\hspace{0.1cm}.
    \end{split}
    \label{eq.85}
\end{equation}

We have slightly abused the notation in the last equality. Since $c_{1,2}$ are real, they are not affected by $\mathcal{I}$. Moreover, it should be understood that the vector field components $(j\partial_x;p\partial_x)$ are sent to quadratic densities via the action $\mathcal{I}\cdot(j\partial_x,p\partial_x)^\top$. 

It is reminiscent of the situation depicted in \cite{KheMis02} for the Virasoro algebra, where it was shown how extra Hamiltonian hierarchies, such as the Camassa-Holm equations \cite{CamHol93} or the Hunter-Saxton equations \cite{HunSax91}, are linked to specific $\mathrm{H}^1$-energies on $\mathfrak{vir}$.

The non-degeneracy nature of the above operator $\mathcal{I}$ depends on $\theta_2-\iota_2\partial_x^2$; the latter is non-degenerate whenever $\theta_2\not=0$, and it has a non-trivial kernel made of constant vector fields on $\mathrm{S}^1$ for $\theta_2=0$.

The $\mathfrak{bms}_3$ bi-Hamiltonian hierarchy we have considered so far happens to be attached to the natural choice: $\theta_1=0,\iota_1=0; \theta_2=1,\iota_2=0$. In fact, then $\mathcal{I}$ is simply the identity operator.
By construction, the frozen point is of the form given in \textbf{Proposition~\ref{prop:proposition4.1}} and it can easily be checked that the resulting constant bracket confuses with the one for $\mathscr{D}$. Equivalently, at the level of the Hamiltonian equations of motion, we have that: 
\begin{proposition}\label{prop:proposition4.2}
    For given smooth functions $\widetilde{\mathcal{F}},\widetilde{\mathcal{G}}\in C^\infty(\mathfrak{bms}_3^*)$, the Lie-Poisson flows on the (regular) dual of $\mathfrak{bms}_3$ are bi-Hamiltonian: 
    \begin{equation}
    \begin{split}
        \frac{d}{dt}\big(u_1dx^2;u_2dx^2\big)&=\mathrm{ad}^*_{\mathfrak{bms}_3}\big(\delta_{u_1}\widetilde{\mathcal{F}}\cdot\partial_x;\hspace{0.1cm}\delta_{u_2}\widetilde{\mathcal{F}}\cdot\partial_x\big)\cdot(u_1dx^2,c_1; u_2dx^2,c_2)\\
        &=\mathrm{ad}^*_{\mathfrak{bms}_3}\big(\delta_{u_1}\widetilde{\mathcal{G}}\cdot\partial_x;\hspace{0.1cm}\delta_{u_2}\widetilde{\mathcal{G}}\cdot\partial_x\big)\cdot\bigg(\frac{1}{2}dx^2,0;\frac{1}{2}dx^2,0\bigg)\hspace{0.1cm}.
    \end{split}
    \label{eq.86}
    \end{equation}
\end{proposition}

A given flow is then considered either \textit{1)} from a first Lie-Poisson Hamiltonian structure on the whole dual of $\mathfrak{bms}_3$, or \textit{2)} from a second compatible Hamiltonian structure on a specific coadjoint orbit.

Of course, the asymptotic flat case is included in this viewpoint (remember that $\kappa_1$ is the parameter of a coboundary related to the central extension with respect to $c_1$). The vanishing of $c_1$ implies to pick $m_0=\big(0,0;1/2,0)$, which thus lead to a compatible Hamiltonian operator of the form 
\begin{equation}
    \mathscr{D}=\begin{pmatrix}
        0 & \partial \\ \partial & 0
    \end{pmatrix}\hspace{0.1cm}.
    \label{eq.87}
\end{equation}

Notice that we recover this result by a shift of $-1/2$ on $j_0$. This is in accordance, in the Gel'fand-Dorfman formalism, with a shift of the cocyle by a coboundary. 

More globally, a $\mathfrak{bms}_3$ bi-Hamiltonian pair should exist with regard to $\mathrm{H}^1_{(\theta_1,\iota_1;\theta_2,\iota_2)}$ for a frozen bracket computed at 
\begin{equation}
    m_0=\bigg(\frac{\theta_1+\kappa_1\theta_2}{2}dx^2, \kappa_1(\iota_1+\iota_2); \hspace{0.1cm} \frac{\theta_2}{2}dx^2, \iota_2\bigg)\hspace{0.1cm}.
    \label{eq.88}
\end{equation}

\begin{remark}\label{rem:remark4.3}
    A detailed study in terms of \textit{Euler equations} \cite{Arn89} is suggested. It would be based on the behavior at the level of the group rather than the (dual) Lie algebra. Given a dynamical system $\frac{dv}{dt}=\mathcal{F}(v)$ on $\mathfrak{g}$, the so-called Euler equation then refers to the description of the evolution of the velocity vector $v$ along a geodesic with respect to a right-invariant metric\footnote{Equivalently a left-invariant one; the difference being a global minus sign in the Euler equation.} of the Lie group $G$. For instance, regarding the KdV equation, it could be characterized as the Euler equation describing the geodesic flow on the Virasoro group with respect to a right-invariant $L^2$-metric (\cite{OvsKhe87,Mis97}). We refer the interested reader to \cite{ArnKhe98} for more details.
\end{remark}

Actually, it is instructive to see how this second compatible Hamiltonian operator is related to the codimension of the chosen coadjoint orbit. The codimension of the orbit going through $m_0$ is usually called the \textit{codimension} of the Lie-Poisson pair $\{,\}_\mathrm{LP}$ and $\{,\}_0$.

Let $\mathcal{W}_m$ and $G_m$ be, respectively, the coadjoint orbit of $m\in\mathfrak{g}^*$ and the little group at $m$ with respect to the coadjoint representation. For the constant 'frozen' brackets, their symplectic leaves are the translations of the tangent space $\mathrm{T}_{m_i}\mathcal{W}_{m_i}$ through the chosen point $m_i$. We formally identify the codimension of the orbit with the dimension of $\mathfrak{g}_{m}=\mathrm{Lie}\big(G_m\big)=\big\{X\in\mathfrak{g}\hspace{0.1cm}|\hspace{0.1cm}\mathrm{ad}^*_\mathfrak{g}X\cdot m=0\big\}$.

For an exceptional semi-direct product $G\ltimes_\mathrm{Ad} \mathfrak{g}_{\mathrm{ab}}$, it is known (\cite{Raw75, Bag98, BarObl15}) that the coadjoint orbits $\mathcal{W}_{(j,p)}$ correspond to the bundle of little group orbits `up to' the cotangent spaces at each point of $\{\mathrm{Ad}^*f\cdot p|\hspace{0.1cm} f\in G\}$.

Denote by $\mathcal{W}_{(j,c_1:p,c_2)}$ the corresponding orbit of $(j,c_1;p,c_2)\in\mathfrak{bms}_3^*$. Recall that the central charges $c_1$ and $c_2$ are invariants on the orbits in $\mathfrak{bms}_3^*$; thus, we restrict ourselves to the orbits in the hyperplane $\big\{(jdx^2,c_1; pdx^2,c_2)|\hspace{0.1cm} c_1=c_1^0; \hspace{0.1cm} c_2=c_2^0\big\}$. With the frozen point $m_0=(1/2,0; 1/2, 0)$, the orbit in this hyperplane is of codimension $2$ (since it is characterized by two distinct vector fields on $\mathrm{S}^1$). Hence, $\mathcal{W}_{(1/2,0; 1/2, 0)}$ is of codimension 4.

Following the Virasoro discussion in \cite{KheMis02}, we believe that the admissible frozen points in $\mathfrak{bms}_3^*$ correspond to the types of $\mathfrak{bms}_3$ coadjoint orbits of codimension 4: 
\begin{conjecture}\label{conj:conjecture4.4}
    For a non-vanishing $c_1$ ($\kappa_1\equiv1$), all the Poisson pairs $\big(\{,\}_\mathrm{LP}^{\mathfrak{bms}_3^*},\{,\}_0^{\mathfrak{bms}_3^*}\big)$ of codimension 4 on $\mathfrak{bms}_3^*$ are classified into fifteen classes depending on the type of the frozen point:
    \begin{equation}
        \begin{split}
            &\bigg(\frac{dx^2}{2},0;\frac{dx^2}{2},0\bigg),\hspace{0.1cm}\bigg(dx^2,0;\frac{dx^2}{2},0\bigg),\hspace{0.1cm} \bigg(\frac{dx^2}{2},1;\frac{dx^2}{2},0\bigg),\hspace{0.1cm}\bigg(dx^2,1;\frac{dx^2}{2},0\bigg);\\
            & \big(0,1;0,1\big),\hspace{0.1cm}\bigg(\frac{dx^2}{2},1;0,1\bigg),\hspace{0.1cm} \big(0,2;0,1\big),\hspace{0.1cm} \bigg(\frac{dx^2}{2},2;0,1\bigg);\\
            & \bigg(\frac{dx^2}{2},1;\frac{dx^2}{2},1\bigg),\hspace{0.1cm} \bigg(dx^2,1;\frac{dx^2}{2},1\bigg),\hspace{0.1cm} \bigg(\frac{dx^2}{2},2;\frac{dx^2}{2},1\bigg),\hspace{0.1cm} \bigg(dx^2,2;\frac{dx^2}{2},1\bigg);\\
            & \bigg(\frac{dx^2}{2},0; 0,0\bigg),\hspace{0.1cm} \big(0,1; 0,0\big),\hspace{0.1cm} \bigg(\frac{dx^2}{2},1; 0,0\bigg)\hspace{0.1cm}.
        \end{split}
        \label{eq.89}
    \end{equation}
\end{conjecture}

It is suggested that the $\mathfrak{bms}_3$ integrable hierarchy uncovered in \textbf{Theorem~\ref{thm:theorem2.10}} is not unique. In future work, we shall address an exhaustive classification of all possible $\mathfrak{bms}_3$-like hierarchies related to general $\mathrm{H}^1_{(\theta_1,\iota_1; \theta_2,\iota_2)}$ energies. Regarding the degenerate situation (for instance $(0,1;0,1)$), we might work on homogeneous spaces. One should also wonder about their integrability character. Additionally, as shown in \cite{Fue.al18}, suitable boundary conditions reduce the Einstein equations to the known $\mathfrak{bms}_3$-PDEs (\icref{eq.54,eq.55} etc). Thus, this conjecture offers a perspective on solving the $3\mathrm{d}$ Einstein equations in specific sectors using integrability methods. 
\\[1em]

Finally, we take a step back. So far, we have covered the whole bottom part of the introductory schematic \bcref{fig:bms3_hierarchy}. Another line of research can be investigated, again based on the KdV scenario. 

We mentioned in the end of \bcref{sec:section1} that the dual of the Virasoro algebra is naturally associated with the space of Schrödinger operators $\big\{c^*\partial_x^2+u(x)\big\}$. It is also known that there is a one-to-one correspondence between the conjugacy classes of the monodromy attached to these operators and the set of coadjoint orbits, so that the trace of the monodromy furnishes a Casimir for the pencil $\{,\}_\mathrm{LP}^{\mathfrak{vir}^*}+\mu\{,\}_0^{\mathfrak{vir}^*}$. It would be tempting to see whether this type of result still holds for $\mathfrak{bms}_3$. A promising first hint lies in the connection between $\mathfrak{bms}_3^*$ and the so-called \textit{energy-dependent Schrödinger operators}. The last section addresses this point. 

\section{An energy-dependent Lax pair}\label{sec:section5}

It is a well-known fact that the KdV hierarchy stems from the Schrödinger operator $\mathcal{L}:=\partial_x^2-u(x)$, where $u\in\mathrm{C}^\infty(\mathbb{R})$. Indeed, one may associate to $\mathcal{L}$ the differential part of its $m^{\mathrm{th}}$-root $\big(\mathcal{L}^{m/2}\big)_+$ in such a way that: \textit{i)} the commutators $\big[\big(\mathcal{L}^{m/2}\big)_+,\mathcal{L}\big]$ yield to the flows of non-linear (but integrable) PDEs and \textit{ii)} the isospectrality of the spectrum for the linear problem $\mathcal{L}\phi(x)=\lambda\phi(x)$ is preserved along these flows. We say that $\big\{\big(\mathcal{L},(\mathcal{L}^{m/2})_+\big),\hspace{0.1cm} m\hspace{0.1cm} \mathrm{odd}\big\}$ are \textit{Lax pairs}\footnote{An even value of $m$ leads to a trivial flow.}.

Specifically for the linear problem, we deform both $\phi(x)\to\phi(x;t)$ and $u(x)\to u(x;t)$ with respect to some $t$-parameter (which should be thought of as an evolutionary parameter) so that $\partial_t\phi=\mathcal{P}\phi$, for $\mathcal{P}$ a differential operator. The \textit{Lax equation} $\partial_t\mathcal{L}=[\mathcal{P},\mathcal{L}]$ expresses the \textit{compatibility}\footnote{\textit{Compatibility} in the sense that a solution $\phi(x;0)$ of the linear problem at $t=0$ is still a solution at time $t$ (after it evolves with respect to $\mathcal{P}$) if, and only if, $u$ solves the Lax equation.} between the spectral problem and the $t$-deformation of the solution. `Isospectrality' then means that the spectrum is independent of $t$.

For the Schrödinger problem, Lax operators are given by $\mathcal{P}=\big(\mathcal{L}^{m/2}\big)_+$, $m\in2\mathbb{N}+1$ and the flows $\partial_{t_m}\mathcal{L}=\big[(\mathcal{L}^{m/2})_+,\mathcal{L}]$ represent the \textit{KdV hierarchy}. In particular, one should bear in mind that these KdV flows identify themselves with the Hamiltonian equations on the (regular) dual of the Virasoro algebra. The literature on the subject is extensive. We refer the interested reader to some of the classic papers \cite{Adl79, DriSok84, Wil79, ZakFad71}, reviews \cite{Man79, Dic03} and references therein.

In this context, we now show that the coadjoint orbits of $\mathfrak{bms}_3$ describe the Lax flows of a specific linear problem.

\begin{definition}\label{def:definition5.1}
    Take $N\geq1$. By a $N$-\textit{energy-dependent operator}, we mean a Schrödinger-like operator of the form
    \begin{equation}
    \mathbb{L}:=\bigg(\sum_0^Na_i\lambda^i\bigg)\partial_x^2+\sum_0^Nu_i(x)\lambda^i\hspace{0.1cm},
    \label{eq.90}
\end{equation}
where $\lambda$ refers to the spectral parameter, the $a_i$ are constants and $u_i$ functions of $x$.    
\end{definition}
We focus on the $N=2$ case and consider the linear spectral problem:
\begin{equation}
    \mathbb{L}\phi=0\hspace{0.1cm}.
    \label{eq.91}
\end{equation}

Integrability aspects of energy-dependent Schrödinger problems have been studied primarily by \cite{Mar80} and \cite{AntFor87-2, AntFor87}. In general, one cannot hope to exhibit a Lax pair by computing fractional powers of such operators, and therefore the previous Gel'fand-Dickey procedure fails. 

Nevertheless, \cite{AntFor89} found an efficient way to still obtain Lax pairs. As before, deform the solution and the potential components in terms of an evolution parameter $t$. We search for functions $\mathcal{A}\big[u_{i}^{(j)};\lambda\big]$ and $\mathcal{B}\big[u_{i}^{(j)};\lambda\big]$ such that the Lax equation $\partial_t\mathbb{L}=\big[\mathbb{P},\mathbb{L}\big]$ is the compatibility condition between \icref{eq.91} and  
\begin{equation}
    \partial_t\phi=\mathbb{P}\phi\equiv\big(\mathcal{A}\partial_x+\mathcal{B}\big)\phi\hspace{0.1cm}.
    \label{eq.92}
\end{equation}

The Lax flows write in terms of familiar $3$rd-order differential operators
\begin{equation}
    \sum_0^2 \frac{\partial u_i}{\partial t}\lambda^i=\bigg(\sum_0^2\mathcal{J}_i\lambda^i\bigg)\mathcal{A} \hspace{0.5cm}\text{with}\hspace{0.5cm} \mathcal{J}_i=\bigg(\frac{1}{2}a_i\partial_x^3+u_i\partial_x+\partial_x\circ u_i\bigg)\hspace{0.1cm},
    \label{eq.93}
\end{equation}
where $\mathcal{B}$ must be proportional to $\mathcal{A}$ in order to cancel an inconvenient first order derivative term. 

It is instructive\footnote{This ansatz can be fully justified via the $r$-matrix formalism applied on the Virasoro-valued loop algebra. For details, see \cite{ForReySem89}.} to look for $\mathcal{A}$ as an \textit{arbitrary} polynomial expansion with respect to the spectral parameter:
\begin{equation}
    \mathcal{A}\big[u_i^{(j)};\lambda \big]\equiv\sum_{i=0}^nA_{n-i}\big[u_i^{(j)}\big]\lambda^i, \hspace{0.1cm} \forall n>0 \hspace{0.5cm} \text{with} \hspace{0.5cm} A_i=0, i<0\hspace{0.1cm}.
    \label{eq.94}
\end{equation}

The above Lax equations \icref{eq.93} split into two sets of equations with respect to powers of $\lambda$ leading to \textit{i)} $3$ non-trivial flows and \textit{ii)} $n$ recursion relations. The latter provide a recipe to compute recursively the $A_i$ coefficients, \textit{up to} $A_{n-1}$. To get the first, we should fix $u_2=\text{Cte}\equiv-1$ so that this flow becomes trivial\footnote{An alternative would be to fix $A_n=0$ so that $u_0=\mathrm{Cte}$. It brings us to the so-called \textit{Harry-Dym hierarchy} \cite{AntFor88}.}. Since it corresponds to the recursion at $n$, the whole construction closes: 
\begin{equation}
    \mathcal{J}_0A_{k-2}+\mathcal{J}_1A_{k-1}+\mathcal{J}_2A_k=0, \hspace{0.1cm} k=0,\cdots,n\hspace{0.1cm}.
    \label{eq.95}
\end{equation}
The constraint equations start with $A_0\in\text{Ker}\mathcal{J}_2$. It is solvable whenever $a_2\equiv0$ \cite{AntFor89}. An infinite sequence of isospectral flows is then derived from this Lax construction.
\begin{proposition}\label{prop:proposition5.2}
    The Lax equations of motion for the $2$-energy-dependent Schrödinger problem take the form
    \begin{equation}
        \begin{pmatrix}
            u_0\\u_1
        \end{pmatrix}_{,t_n}=\begin{pmatrix}
            0 & \mathcal{J}_0 \\ \mathcal{J}_0 & \mathcal{J}_1
        \end{pmatrix}\begin{pmatrix}
            A_{n-1}\\ A_n
        \end{pmatrix}\hspace{0.1cm},
        \label{eq.96}
    \end{equation}
    and the associated matrix differential operators is Hamiltonian.
\end{proposition}

One recognizes precisely that the matrix above is similar to the first Hamiltonian operator $\mathscr{E}$ (see \icref{eq.52}) that we found working on the variational complex. 

It is worth mentioning that whenever $a_1=0$ and $u_1=0$, we do recover the standard KdV hierarchy. 

The next result unravels the connection with the $\mathfrak{bms}_3$ algebra.
\begin{theorem}\label{thm:theorem5.3}
    The Lax flows for the $2$-energy-dependent Schrödinger problem are the Hamiltonian equations on the (regular) dual of $\mathfrak{bms}_3$
    \begin{equation}
        \big(u_1,u_0\big)_{,t_n}=\mathrm{ad}^*_{\mathfrak{bms}_3}\bigg(A_n=\delta_{u_1}\widetilde{\mathcal{H}}_n;A_{n-1}=\delta_{u_0}\widetilde{\mathcal{H}}_n\bigg).\bigg(u_1,-\frac{a_1}{2};u_0,-\frac{a_0}{2}\bigg)\hspace{0.1cm},
        \label{eq.97}
    \end{equation}
    for some smooth function $\widetilde{\mathcal{H}}_n$.
\end{theorem}

Actually, two extra Hamiltonian structures can be extracted from the constraint recursive relations of the Lax pair \cite{AntFor89}. These are, respectively, 
\begin{equation}
    \mathscr{F}=\begin{pmatrix}
        \mathcal{J}_0 & 0 \\ 0 & -\mathcal{J}_2
    \end{pmatrix} \hspace{0.2cm} \text{and} \hspace{0.2cm} \mathscr{G}=\begin{pmatrix}
        -\mathcal{J}_1 & -\mathcal{J}_2 \\ -\mathcal{J}_2 & 0
    \end{pmatrix}\hspace{0.1cm}.
    \label{eq.98}
\end{equation}

The triplet $\big\{\mathscr{E},\mathscr{F},\mathscr{G}\big\}$ forms a \textit{tri-Hamiltonian} structure. Since $\mathscr{G}$ is non-degenerate, $\mathscr{R}:=\mathscr{F}\circ\mathscr{G}^{-1}$ is a well-defined operator, and we can conclude to the existence of an infinite sequence of functionals $\widetilde{\mathcal{H}}_n$ such that: 
\begin{equation}
    \big(u_0,u_1\big)^\top_{,t_n}=\mathscr{E}\mathbf{d}\widetilde{\mathcal{H}}_n=\mathscr{F}\mathbf{d}\widetilde{\mathcal{H}}_{n+1}=\mathscr{G}\mathbf{d}\widetilde{\mathcal{H}}_{n+2}\hspace{0.1cm}.
    \label{eq.99}
\end{equation}

Hence, it appears that the $2$-energy-dependent hierarchy and the $\mathfrak{bms}_3$ hierarchy do not fully coincide. The latter may be a reduction of the former.
 
In fact, these last results on energy-dependent Lax pairs (in the `Antonowicz-Fordy' setting) might be completely understood from an \textit{r-matrix} perspective. 

Let $\mathcal{L}(\mathfrak{vir})=\oplus_{i\in\mathbb{Z}}\mathfrak{vir}\lambda^i$ be the Virasoro-valued loop algebra. Such algebras offer an arena to work with $r$-matrices since they present a canonical decomposition between subalgebras of negative and positive powers in Laurent series \cite{Sem08}. The mapping $r:=P_+-P_-$, where $P_{\pm}$ denotes the projection along $\mathcal{L}(\mathfrak{vir})_\pm$, is a natural choice of $r$-matrix. For the $2$-energy dependent case, \cite{ReySem89} showed\footnote{Actually, they did it for any $N$-energy dependent Schrödinger operators.} that there exists $3$ compatible brackets $\{,\}_{r\lambda^k}$, $k\in\{0,1,2\}$, given in terms of the coefficients $(u_0+u_1\lambda-\lambda^2, \hspace{0.1cm}a_0+a_1\lambda)$, which identify with the previous tri-Hamiltonian structure: 
\begin{equation}
    \{,\}_{r\lambda^0}\equiv\{,\}_\mathscr{G} \hspace{0.2cm} \text{,} \hspace{0.2cm} \{,\}_{r\lambda^1}\equiv\{,\}_\mathscr{F} \hspace{0.2cm} \text{,} \hspace{0.2cm} \{,\}_{r\lambda^2}\equiv\{,\}_\mathscr{E}\hspace{0.1cm} .
    \label{eq.100}
\end{equation}

Additionally, exceptional semi-direct structures may be described from graded algebras. A noteworthy fact is that $\lambda^2\mathcal{L}(\mathfrak{vir})_+$ is an ideal of $\mathcal{L}(\mathfrak{vir)}_+$ with respect to the $r$-bracket, and it happens that the quotient precisely recovers the $\mathfrak{bms}_3$ algebra: 
\begin{equation}
    \mathcal{L}(\mathfrak{vir})_+/_r\lambda^2\mathcal{L}(\mathfrak{vir})_+\simeq\mathfrak{bms}_3\hspace{0.1cm}.
    \label{eq.101}
\end{equation}

Based on this observation, we expect to soon provide a fully detailed description of the $\mathfrak{bms}_3$ integrable hierarchy (c.f. \textbf{Theorem~\ref{thm:theorem2.10}}) using the $r$-matrix formalism.

\newpage

\section*{Conclusions and Perspectives}\label{sec:conclusion}
\phantomsection
We conclude by enumerating the main results of this work.

\textbf{\textit{1)}} We have revisited the proof of the integrability of the hierarchy based on the extended Lie algebra of BMS$_3$, introduced in \cite{Fue.al18}, as possessing a bi-Hamiltonian structure. Specifically, we have established an `algebraic' and a `geometrical' proof:

\textit{a)} Working on the variational complex, we have seen how Gel'fand and Dorfman results on Hamiltonian structures could be used on the centrally extended $\mathfrak{bms}_3$ to get a compatible pair of such. We were able to obtain a well-defined Nijenhuis operator and justify the existence of an infinite number of independent flows. 

\textit{b)} On a Lie-Poisson perspective, starting from the Lie-Poisson bracket on the dual of $\mathfrak{bms}_3$, we recovered this bi-Hamiltonian structure thanks to the "frozen" method, by looking at the previous bracket on a well-suited point in the dual. 

\textbf{\textit{2)}} More interestingly, from the non-uniqueness of the 'frozen' point, we conjecture that there are other $\mathfrak{bms}_3$-like Hamiltonian pairs. Bearing in mind one of the key results of \cite{Fue.al18}, in connection with solutions of $3$d Einstein equations, we then wonder whether it extends to these new hierarchies.

\textbf{\textit{3)}} Alternatively, we have depicted the construction of a $\tau$-scheme for the asymptotically flat scenario.  

\textbf{\textit{4)}} An integrable bi-Hamiltonian structure was also exhibited for the centrally extended $\widehat{\mathfrak{ads}}_3$ algebra. It consisted mainly of two copies of the KdV structure, as we could have expected from the Lie algebra definition. The `flat limit' argument allowed us to recover the $\mathfrak{bms}_3$ bi-Hamiltonian structure from the one of $\widehat{\mathfrak{ads}}_3$.

\textbf{\textit{5)}} Based on well-studied links between Lie-Poisson structures, Schrödinger operators, and exact pencils, we started to glimpse a connection between $\mathfrak{bms}_3$ integrable features and a subclass of energy-dependent Schrödinger operators, even if more work needs to be done in this direction. 
\\[5mm]
In the view of the previously stated results, we consider at least four research perspectives that we should naturally address in the future. We list them below: 

\underline{\textit{i)}} To provide a complete $r$-matrix description of the $\mathfrak{bms}_3$ integrable system here discussed.

\underline{\textit{ii)}} To investigate the $\mathfrak{bms}_3$ mastersymmetries. This also suggests that a \textit{`matrix-$\mathfrak{bms}_3$'} integrable hierarchy should exist, from which the $\mathfrak{bms}_3$ hierarchy we have considered would be a reduced version (in a similar fashion to what is done in KdV, see \cite{CasMagPed93}). This is akin to the `Sato's perspective' on soliton equations.

\underline{\textit{iii)}} To address the situation of the Euler equation on the BMS$_3$ group together with the classification of all admissible $\mathfrak{bms}_3$-like hierarchies. 

\underline{\textit{iv)}} To look into the existence of a $\mathfrak{bms}_4$ bi-Hamiltonian (and potentially integrable) structure. The last idea mentioned in \textit{iii)} should be tackled in this case too. 

\section*{Acknowledgments}

The author warmly thanks his advisors José Luis Jaramillo, Michele Lenzi and Carlos F. Sopuerta, and is grateful to Maxime Fairon and Mikhail Semenov-Tian-Shansky for their many comments and guidance. The author is also obliged to Glenn Barnich, Cristina Câmara and Gabriel Cardoso for insightful discussions.
This work has been achieved in the frame of the EIPHI Graduate school (contract "ANR-17-EURE-0002"), and in particular through the EIPHI-2025 project "Resonan-SIES". We also thank the Erwin Schrödinger International Institute for Mathematics and Physics
(ESI) where part of this work has been developed during the Workshop on “Hyperboloidal Foliations and their Application”.

\section*{Appendix: Proof of \textbf{Proposition~\ref{prop:proposition2.8}}}\label{sec:appendix}
\phantomsection
\begin{proof}
 The proof relies essentially on a series of computations. We proceed by induction on $n$. Assume that $\overline{\mathcal{K}}_n=\mathscr{D}\xi_n$ for some $\xi_n\in\widetilde{\Omega}^1$. We may write it $\xi_n=\xi_{n,1}\mathbf{d}u_1+\xi_{n,2}\mathbf{d}u_2$, $\xi_{n,i}\in\mathscr{A}$. We need to prove that $\overline{\mathcal{K}}_{n+1}=\mathscr{R}\overline{\mathcal{K}}_n$ lies in the image of $\mathscr{D}$. To fix notations, let us say $\overline{\mathcal{K}}_n=\mathcal{K}_{n,1}\partial_{u_1}+\mathcal{K}_{n,2}\partial_{u_2}$ and call $r_2:=\mathcal{J}_2\partial^{-1}$ and $r_1:=\mathcal{J}_1\partial^{-1}-\mathcal{J}_2\partial^{-1}$, so that: 
\begin{equation}
    \mathscr{R}=\begin{pmatrix}
        r_2 & r_1 \\ 0 & r_2
    \end{pmatrix}\hspace{0.1cm}.
    \label{Eq.90}
\end{equation}

A computation gives: 
\begin{equation}
        \overline{\mathcal{K}}_{n+1}=\mathscr{D}
        \begin{pmatrix}
        -c_2\partial\mathcal{K}_{n,2}+u_2\partial^{-1}\mathcal{K}_{n,2}+\partial^{-1}\circ u_2\mathcal{K}_{n,2} \\ \\
        \begin{split}
            -c_2\partial\mathcal{K}_{n,1}&+u_2\partial^{-1}\mathcal{K}_{n,1}+\partial^{-1}\circ u_2\mathcal{K}_{n,1}+\big(2c_2-c_1\big)\partial\mathcal{K}_{n,2} \\
            &+\big(u_1-2u_2\big)\partial^{-1}\mathcal{K}_{n,2}+\partial^{-1}\circ\big(u_1-2u_2\big)\mathcal{K}_{n,2}
        \end{split}
        \end{pmatrix}\hspace{0.1cm}.
        \label{Eq.91}
\end{equation}

If one manages to check that $u_2\mathcal{K}_{n,2}$ and $u_2\mathcal{K}_{n,1}+(u_1-2u_2)\mathcal{K}_{n,2}$ can be written as a total derivative with respect to $\partial$, the induction will be complete. By assumption, 
\begin{equation}
    \begin{split}
    \mathcal{K}_{n,1}&=r_2^n\big(u_1^{(1)}+u_2^{(1)}\big)+\sum_{k=0}^{n-1}r_2^{n-1-k}r_1r_2^{k}u_2^{(1)}\hspace{0.1cm}, \\
    \mathcal{K}_{n,2}&=r_2^nu_2^{(1)}\hspace{0.1cm}.
    \end{split}
    \label{Eq.92}
\end{equation}

Performing successive integration by parts, we first check that 
\begin{equation}
    u_2\mathcal{K}_{n,2}=u_2^{(1)}\big(r_2^*\big)^nu_2+\partial\mathcal{A}_{n,1}=u_2^{(1)}\partial^{-1}\mathcal{K}_{n,2}+\partial\mathcal{A}_{n,1}=-u_2\mathcal{K}_{n,2}+\partial\big(\mathcal{A}_{n,1}+\mathcal{B}_{n,1}\big), 
    \label{Eq.93}
\end{equation}
for $\mathcal{A}_{n,1}$ and $\mathcal{B}_{n,1}\in\mathscr{A}$. Thus, $u_2\mathcal{K}_{n,2}=\partial\mathcal{S}_{n,1}$ where $\mathcal{S}_{n,1}=\frac{1}{2}\big(\mathcal{A}_{n,1}+\mathcal{B}_{n,1}\big)\in\mathscr{A}$. 

Similarly with $\mathcal{A}_{n,2}\cdots\mathcal{E}_{n,2}\in\mathscr{A}$, 
\begin{equation}
    \begin{split}
        \big(u_1-2u_2\big)\mathcal{K}_{n,2}+u_2\mathcal{K}_{n,1}&=-u_2r_2^nu_1^{(1)}-u_1r_2^nu_2^{(1)}+u_2r_2^nu_2^{(1)}
        +u_2^{(1)}\bigg(\sum_{k=0}^{n-1}r_2^{n-1-k}r_1r_2^k\bigg)^*u_2\\ &+\partial\big(\mathcal{A}_{n,2}+\mathcal{B}_{n,2}+\mathcal{C}_{n,2}+\mathcal{D}_{n,2}+\mathcal{E}_{n,2}\big)\hspace{0.1cm}.
    \end{split}
    \label{Eq.94}
\end{equation}

Working out the remaining adjoint term, we are led to
\begin{equation}
    u_2^{(1)}\bigg(\sum_{k=0}^{n-1}r_2^{n-1-k}r_1r_2^k\bigg)^*u_2=-u_2\sum_{k=0}^{n-1}r_2^kr_1r_2^{n-1-k}u_2^{(1)}+\partial\mathcal{F}_{n,2}, \hspace{0.2cm} \mathcal{F}_{n,2}\in\mathscr{A}\hspace{0.1cm}.
    \label{Eq.95}
\end{equation}
such that $\big(u_1-2u_2\big)\mathcal{K}_{n,2}+u_2\mathcal{K}_{n,1}=\frac{1}{2}\partial\big(\mathcal{A}_{n,2}+\cdots+\mathcal{F}_{n,2}\big)\equiv\partial\mathcal{S}_{n,2}$, $\mathcal{S}_{n,2}\in\mathscr{A}$.

\end{proof}

\bibliographystyle{alpha}
\bibliography{bibliography}

\end{document}